\documentclass{article}
\usepackage[usenames,dvipsnames,svgnames,table]{xcolor}
\usepackage{latexsym, amscd, amsfonts, eucal, mathrsfs, amsmath, amssymb, amsthm, xypic, makecell, boldline, float}
\usepackage{tikz}
\usepackage{subfig}
\usepackage{graphicx}
\usepackage{setspace}
\usepackage[titletoc]{appendix}
\usepackage{hyperref}
\hypersetup{
    colorlinks=true,       
    linkcolor=Black,          
    citecolor=Black,       
    filecolor=Black,      
    urlcolor=Black          
}

\usepackage[round]{natbib}

\newcommand\addtag{\refstepcounter{equation}\tag{\theequation}}

\usetikzlibrary{calc}
\tikzset{
    solid node/.style={circle,draw,inner sep=1.5,fill=black},
    hollow node/.style={circle,draw,inner sep=1.5}}

\def\BR{\mathbb{R}}

\newtheorem{thm}{Theorem}
\newtheorem{dfn}{Definition}
\newtheorem{lma}{Lemma}
\newtheorem{prop}{Proposition}

\begin{document}

\title{Contracts for acquiring information}
\author{Aubrey Clark and Giovanni Reggiani \footnote{We are grateful to our PhD advisors Eric Maskin, Oliver Hart, Mihai Manea, Bengt Holmstrom, Bob Gibbons, and Abhijit Banerjee. For comments on various versions of the paper we thank Gabriel Carroll, Ben H{\'e}bert, Bengt Holmstr{\"o}m, Ryota Iijima, Divya Kirti, Jonathan Libgober, Xiaosheng Mu, Jeff Picel, Gleb Romanyuk, Jann Spiess, Eduard Talam{\`a}s, Yao Zeng, and audiences from the Theory and Contracts lunches at Harvard and MIT. Email: aubs.bc@gmail.com} }

\date{\today}

\maketitle

\begin{abstract}
Information is costly for an agent to acquire and unobservable to a principal. We show that any Pareto optimal contract has a decomposition into a fraction of output, a state-dependent transfer, and an optimal distortion. Under this decomposition: 1) the fraction of output paid is increasing in the set of experiments available to the agent, 2) the state-dependent transfer indexes contract payments to account for differences in output between states, 3) the optimal distortion exploits complementarities in the cost of information acquisition: experiment probabilities unalterable via contract payments stuck against liability limits are substituted for, the substitution occurring according to complementarities in the cost of information acquisition, and 4) if and only if the agent's cost of experimentation is mutual information, the optimal distortion takes the form of a decision-dependent transfer.
\end{abstract}

\section{Introduction}

Desiring to learn the state of the world a principal hires an expert agent to generate data. Sadly, the principal is not able to observe the agent's selection of data generating process and so ties compensation to verifiable outputs correlated with this selection.

As an example, consider the partnership between someone deciding on a portfolio of stocks---the principal---and a financial advisor---the agent. A compensation scheme could depend on the suggested portfolio and the terminal stock prices.

This paper describes the form of the optimal compensation scheme. We show that it has an intuitive structure: a fraction of output, a state-dependent transfer, and an optimal distortion

 When the cost is mutual information, the optimal distortion only depends on the decision the agent takes. The resulting contract can be described simply: stipulation of the liability limits, a punishment or reward for each decision based on its ex-ante risk of causing bankruptcy, an indexing payment so the agent is not rewarded purely for finding himself in a lucky state, and a piece rate based on the agent's information acquiring capacity.
 
 For cost functions that are given by the average distance between prior and posteriors, the optimal distortion can vary by state; its form depending on complementarities in the cost of acquiring information about different states of the world. If a contract payment in some state is constrained by a liability limit, then payments in complementary states are instead used to influence the agent.

To illustrate the model, consider the already mentioned example of a risk-neutral investor and a risk-neutral financial advisor. The investor must decide whether to buy a risky (V) or safe (S) asset. The price of the risky asset falls in state $\theta_F$ and rises in state $\theta_R$.

\begin{figure}[h]
\centering
\begin{tabular}{ r !{\vrule width0.9pt}   c  c }
 $y(d,\theta)$ &  $\theta_F$ &  $\theta_R$  \\ \Xhline{0.9pt}
 $V$  & 0 & 10   \\
 $S$  & 5 & 5  
\end{tabular}
\end{figure}

Take the prior probability of state $\theta_F$ to be $2/3$. Absent evidence of state $\theta_F$ the investor will wish to invest in the safe asset. The possibility of gains from trade arises if there is a financial advisor who can acquire information about the likelihood of state $\theta_F$. Information comes in the form of the observation of a random variable whose outcome is correlated with the state. In general we can label signal realizations with the decisions they make optimal, and so an experiment is described by probabilities $p(d|\theta)$ for $d = V,S$ and $\theta = \theta_F, \theta_R$. The principal incurs an experiment cost $c(p)$ and can only select experiments with cost less than $k>0$. Assuming that the state and decision are verifiable, a contract is described by payments $b(d, \theta)$ for $d = V,S$ and $\theta = \theta_F, \theta_R$. 

Consider the agent optimal contract subject to the agent receiving his outside option $0$. Without constraints placed on the available contracts, the usual result is that the investor ``sells the firm" to the financial advisor for a fee. This applies in our model as well, but is modified by the agent's limited capacity $k$---the financial advisor might exert maximal effort after being sold only a fraction $\alpha$ of the firm, and in such a case it is unnecessary to sell the entire firm.

If the financial advisor is protected by limited liability, so that contract payments less than 0 cannot be enforced, selling the firm---or a fraction of the firm---for a fixed fee will not be feasible as it will result in negative contract payments if the financial advisor purchases the risky asset $V$ and prices fall $\theta_F$. In this case, the optimal contract will exploit the agent's incentives not being altered by state-dependent transfers $\beta(\theta)$. In particular, the state-dependent transfer will be set so that the minimum payment in each state is $0$. This extracts maximum profit from the financial advisor without distorting his incentives for acquiring information.

Our main result shows that the optimal contract has the form $b(d,\theta) = \alpha y(d,\theta) - \beta(\theta) - \gamma(d,\theta)$ where the fraction $\alpha$ is chosen, as in the case without liability limits, so that the agent only just exhausts his capacity $k$ for acquiring information, the state-dependent transfer $\beta$ is chosen so that the minimum contract payments in each state touch the liability limits, and the optimal distortion $\gamma$
\[\gamma(d,\theta) = -\frac{1}{\pi(\theta)} \sum_{ d',\theta' } \frac{\partial^2 c(p)}{\partial p(d|\theta) \partial p(d'|\theta')} \frac{\lambda[d',\theta']}{\pi(\theta')},\] 
where $\lambda[d,\theta]$ is a non-negative dual variable for the constraint $0 \leq b(d,\theta)$. The optimal distortion recruits contract payments that are not bound against liability limits in order to influence the signal probabilities for contract payments that are. The extent to which this recruitment happens is modulated by the complementarity in the cost of acquiring different signals in different states and the value of the dual variable, which measures how costly a binding liability limit is to the principal and thus how important it is for him to find an alternative route to affecting the corresponding signal probability.

\paragraph{Related literature}

The first study of optimal incentives for information acquisition was completed by ~\cite{demski1987delegated}. They focused on optimal risk-sharing arrangements and on contracts that depend solely on output. Their main concern was the interaction between the ``planning phase'' (acquiring information) and the ``implementation phase'' (choosing a decision). Their basic point was that, even if implementation is costless (as it is in our model), the unverifiable nature of planning induces moral hazard in implementation since the outcome of the agent's decision provides information about his planning.

\cite{osband1989optimal} was led to the problem from his study of proper scoring rules---incentive schemes that elicit true beliefs (see~\cite{savage1971elicitation}). His goal was to characterize the best proper scoring rule eliciting a forecaster's estimate of the mean of a random variable. The forecaster draws observations of the random variable at a cost, and this cost as well as the number of observations the forecaster makes are unknown to the principal. His focus was on the merits of (a) screening by offering a menu of contracts versus (b) pitting forecasters in competition against one another.

More recently, ~\cite{zermeno2011principal,zermeno2012role} has developed a very general model of incentivizing information acquisition under liability limits. His analysis focuses on the interaction between the verifiable data on which a contract can depend (for instance, the decision may be verifiable but the state may not be) and: (1) the usefulness of menus of contracts, and (2) whether the principal or agent is tasked to take the decision. Menus are useful because they allow contract payments to partially depend on the state: following receipt of a signal the agent gets to choose a decision and a contract from the menu; since the signal is correlated with the state the menu functions as a tool for correlating contract payments with the state---a role that is otherwise forced upon the agent's choice of decision---with a more precise signal leading to better correlation. In general, menus do not allow perfect correlation between contract payments and the state, and for this reason decision making is usually ex-post inefficient---it is distorted so that the agent is more likely to take a decision that reveals the state. Regarding (2), the allocation of decision making authority is irrelevant when the principal and agent can infer from the verifiable data which decision was taken; when this is unclear, whether the principal or agent is tasked to take the decision affects the set of implementable outcomes.

\cite{Carroll16} builds on and refines this model to study the form of ``robust'' incentives for information acquisition: a principal does not know all experiments or experiment costs available to an agent and ranks contracts according to their minimum expected payoff among all experiments and experiment costs including a known set. He shows that the optimal contract is a \emph{restricted investment contract}: the set of decisions available to the agent is restricted, and for unrestricted decisions the optimal contract pays a fraction of output and a state-dependent transfer.


\section{Model}\label{framework}

There are two dates. At date $0$, a principal faces a choice from a finite set of \emph{decisions} $D$. The date $1$ payoff to the principal from a decision depends on the realization of one among a finite set of \emph{states} $\Theta$, as described by an \emph{output function} $y:D \times \Theta \to \BR$. The principal hires an agent to acquire information about the state and take a decision on his behalf but cannot either monitor how much care the agent takes acquiring information or observe the information the agent does acquire. The principal motivates the agent with a \emph{contract} $b:D \times \Theta \to \BR$ that depends on the decision the agent takes and the realized state of the world---both are verifiable.

At date $0$, the principal and agent agree to a contract $b$. At date $1$, before choosing a decision, the agent performs an \emph{experiment} $(X,\{p(\cdot| \theta)\}_{\theta \in \Theta})$, which consists of a finite set of \emph{signals} $X$ and a collection of probability distributions $\{p(\cdot|\theta)\}_{\theta \in \Theta} \subseteq \Delta(X)$ on the set of signals. A \emph{decision rule} $f:X\to \Delta(D)$ describes the agent's randomization over decisions upon observing each signal. 

An \emph{action profile} $(b,X,\{p(\cdot|\theta)\}_{\theta\in \Theta},f)$ consists of a contract, experiment, and decision rule. The agent has preferences over action profiles $$ (b,X,\{p(\cdot|\theta)\}_{\theta\in \Theta},f) $$ given by the expected value of contract payments less an experiment cost: 
$$
\sum_{\theta\in \Theta}\pi(\theta)\sum_{x\in X}p(x|\theta)\sum_{d\in D}f(d|x)b(d,\theta)-c(X,\{p(\cdot|\theta)\}_{\theta\in \Theta}).
$$
The distribution $\pi\in \Delta(\Theta)$ is the agent's \emph{prior belief} over states and the function $c$ is the agent's \emph{cost function}. For now, the cost function is only assumed to respect the ordering of experiments given by informativeness (see Blackwell's theorem\footnote{Experiment $(X,\{p(\cdot|\theta)\}_{\theta\in \Theta})$ is \emph{more informative} than experiment $(X',\{p'(\cdot|\theta)\}_{\theta\in \Theta})$ if for any prior $\pi$ and contract $b$, whenever some expected payoff can be achieved under the latter experiment, then it can be achieved under the former. That is, for each decision rule $f'$ on $X'$ there exists a decision rule $f$ on $X$ such that 
\[\sum_{\theta\in \Theta}\pi(\theta)\sum_{x'\in X'}p'(x'|\theta)\sum_{d\in D}f'(d|x')b(d,\theta) = \sum_{\theta\in \Theta}\pi(\theta)\sum_{x\in X}p(x|\theta)\sum_{d\in D}f(d|x)b(d,\theta).\]

Experiment $(X,\{p(\cdot|\theta)\}_{\theta\in \Theta})$ is \emph{sufficient} for experiment $(X',\{p'(\cdot|\theta)\}_{\theta\in \Theta})$ if there is a \emph{garbling function} $g:X \to \Delta(X')$ such that 
\[p'(x'|\theta) = \sum_{x\in X}p(x|\theta)g(x'|x).\] 
That is, experiment $(X',\{p'(\cdot|\theta)\}_{\theta\in \Theta})$ can be performed by randomizing over its outcomes according to $g$ where the randomization depends on the outcome of the experiment $(X,\{p(\cdot|\theta)\}_{\theta\in \Theta})$.

\emph{Blackwell's theorem} \cite{blackwell1953equivalent} states that an experiment is more informative than  another if and only if it is sufficient.\label{footnote1}}).

The principal has preferences over action profiles $(b,X,\{p(\cdot|\theta)\}_{\theta\in \Theta},f)$ given by the expected value of output less contract payments:
\begin{equation*}
\sum_{\theta\in \Theta}\pi(\theta)\sum_{x\in X}p(x|\theta)\sum_{d\in D}f(d|x)\left(y(d,\theta)-b(d,\theta)\right).
\end{equation*} 

Contracts are restricted to belong to a set $\mathscr{B}$ of \emph{feasible contracts} and experiments are restricted to belong to a set $\mathscr{E}$ of \emph{feasible experiments}. 

A contract $b$ satisfies the \emph{limited liability} constraints if $0\leq b(d,\theta)\leq y(d,\theta)$ for all $d\in D,\theta\in \Theta$. The set of feasible contracts are those satisfying the limited liability constraints. This may mean that the principal and agent are wealth constrained (with wealth normalized to zero or included in output $y$) or that the contract is subject to statutory liability limits (as is the case for a limited liability corporation).

In addition to the incremental cost of acquiring information represented by the cost function $c$, we assume that the agent faces a \emph{capacity constraint} $c(X,\{p(\cdot|\theta)\}_{\theta\in\Theta})\leq k$ that limits the experiments he can perform. The set of feasible experiments are those satisfying this constraint.

An action profile $(b,X,\{p(\cdot|\theta)\}_{\theta \in \Theta},f)$ is \emph{feasible} if the contract $b$ and experiment $(X,\{p(\cdot|\theta)\}_{\theta \in \Theta})$ are feasible, and there is no action profile $$ (b,X',\{p'(\cdot|\theta)\}_{\theta\in \Theta},f'), $$ consisting of the same contract and some feasible experiment, that the agent prefers. 

An action profile can be \emph{improved upon} if there is a feasible profile that makes either the agent or principal better off without making the other worse off. An action profile is \emph{Pareto optimal} if it is feasible and cannot be improved upon.

\paragraph{Normalizing the experiment and decision rule.}\label{normalize}

A decision rule is a garbling function (see footnote~\ref{footnote1}). Thus, given action profile $(b,X,\{p(\cdot| \theta)\}_{\theta\in \Theta},f)$, we may define a garbled experiment $ (D,\{p'(\cdot|\theta)\}_{\theta\in \Theta}) $ by $$ p'(d|\theta) = \sum_{x\in X}p(x|\theta)f(d|x) $$ and choose the decision rule $f':D\to \Delta(D)$ that maps each decision to the degenerate distribution on that decision. The contract has the same expected payoff under this experiment and decision rule and since the former experiment is sufficient for the latter it is weakly more costly. Thus, it is without loss to set the agent's set of signals to be $D$ and the agent's decision rule $f: D\to \Delta(D)$ to map each decision to the degenerate distribution on that decision.

Given this normalization, we will write $p$ for the experiment $(D,\{p(\cdot | \theta)\}_{\theta\in \Theta})$, and write $(b,p)$ for the action profile $(b, D,\{p(\cdot | \theta)\}_{\theta\in \Theta},f)$ in which $f$ maps each decision to the degenerate distribution on that decision. We will write $E_p[b]-c(p)$ and $E_p[y-b]$ for the agent's and principal's utility from action profile $(b,p)$. For experiment $p$, we will write $p(d,\theta)$ for the joint probability $p(d|\theta)\pi(\theta)$, $p(d)$ for the total probability $\sum_\theta \pi(\theta)p(d|\theta)$, and $p(\theta|d)$ for the posterior probability $p(d|\theta)\pi(\theta)/\sum_{\theta'} \pi(\theta')p(d|\theta')$.

\section{Results}\label{Results}

\paragraph{Summary of results}

Since both principal and agent are risk-neutral, a Pareto optimal contract is one that maximizes welfare---expected output less the cost of the agent's chosen experiment---subject to the agent obtaining a given level of utility. A way to characterize all Pareto optimal contracts is to start with the contract equal to output and to then consider feasible alterations that increase the principal's payoff at minimal loss to welfare. Our setup provides two ways to do this without altering the agent's behavior, and thus without any loss in welfare. 

First, since the agent is risk-neutral and does not control the state, altering contract payments by a state-dependent transfer does not alter his behavior. Second, if the agent's capacity constraint binds then a slight scaling down of the contract---multiplying it by a number slightly less than one---will not alter the agent's behavior provided the capacity constraint continues to bind.


The interval of utilities achievable for the agent by applying these maneuvers to output, correspond to the points on the Pareto frontier in which welfare is maximized. Proposition~\ref{prop2} characterizes these contracts. For lower agent utilities there is no way to increase the principal's payoff without reducing welfare, and thus for such contracts the agent's liability limit binds for at least one decision in each state and the agent's capacity constraint does not bind.

We characterize these second-best Pareto optimal contracts in Proposition~\ref{prop3}. These contracts share the features of first-best contracts in that they pay a fraction of output less a state-dependent transfer. But they also involve an optimal distortion, which we give a formula for in this proposition.

In section~\ref{poscomp}, we provide a simplified version of this formula specialized to the class of posterior separable cost functions studied in~\cite{caplin2013behavioral} and~\cite{HebertWoodford16}. Expected reduction in Shannon entropy belongs to this class. Under this cost function the optimal distortion reduces to a decision-dependent transfer in which the agent is punished for taking decisions that are likely to make his liability limits bind and rewarded for decisions that are likely to make the principal's liability limits bind.




\paragraph{First-best contracts}\label{First-best contracts}

In the standard principal-agent setup with risk-neutral parties the \emph{first-best} contracts (those that maximize welfare $E_p[y]-c(p)$) are given by output less a constant, $b=y-t$. If the principal ``sells the firm'' to the agent for a fee $t$, then the agent necessarily maximizes welfare. In problems of information acquisition there are additional first-best contracts. Given that the agent is risk-neutral and does not control the state, modifying a contract by a state-dependent transfer does not alter his incentives. Thus, any $y-\beta$ with $\beta:\Theta \to \BR$ is first-best.

Introducing a capacity constraint into the problem introduces more first-best contracts. There is $\alpha' \in [0,1]$ such that for any first-best contract $y-\beta$ all contracts $\alpha y - \beta$ with $\alpha\in [\alpha',1]$ are also first-best. When the agent's capacity constraint does not bind under contract $y-\beta$, then $\alpha'=1$; if it does bind, $\alpha'$ is less than $1$.

Conversely, if the cost function, viewed as a function on $\BR^{|D \times \Theta|}$, is differentiable, and each first-best experiment assigns positive probability to each decision in each state, then all first-best contracts take this form. 

\begin{prop}\label{prop2}
Define $\alpha'$ by
\[\alpha' = \sup\left\{\alpha \in [0,1] : c(p)< k \text{ for all feasible } p \text{ optimal to contract } \alpha y \right\}.\]
If a contract has the form $\alpha y - \beta$, with $\alpha \geq \alpha'$ and $\beta: \Theta \to \BR$, then it is first-best.

Conversely, suppose the cost function is differentiable and all first-best experiments assign positive probability to each decision in each state. Then, if $b$ is a first best contract, $b = \alpha y - \beta$ for some $\alpha \geq \alpha'$ and $\beta:\Theta \to \BR$.
\end{prop}

\begin{proof}

Let $p$ be an experiment that is optimal for the agent against contract $y$. First, consider a contract  $\alpha y - \beta$ for $\alpha \geq \alpha'$ and $\beta:\Theta \to \BR$ and let $p_\alpha$ be optimal for this contract. We will show that
\begin{equation}\label{ineq3}
E_{p}[y] - E_{p_\alpha}[y] \geq c(p) - c(p_\alpha) \geq \alpha(E_{p}[y] - E_{p_\alpha}[y]).
\end{equation}
Then, if $\alpha \geq \alpha'$, we may assume $c(p_\alpha) = k$, so that $c(p)=k$ and $E_p[y] = E_{p_\alpha}[y]$. This implies that $\alpha y - \beta$ is first-best.

Inequality (\ref{ineq3}) may be proved by noting that $E_{p}[y] - c(p) \geq  E_{p_\alpha}[y] - c(p_\alpha)$ and $E_{p_\alpha}[\alpha y - \beta] - c(p_\alpha) \geq E_{p}[\alpha y - \beta] - c(p)$. Rearranging these inequalities gives the first and second inequalities of (\ref{ineq3}).

Conversely, let $b$ be a first-best contract and $p'$ an experiment that is optimal for this contract. Then the gradient of the agent's objective function at $p'$ under contract $y$ and contract $b$ belong to the normal cone $N_{\mathscr{E}}(p')$ (see Theorem~\ref{wetrock6.12}). 

Since the cost function is differentiable, the set of feasible experiments is regular at $p'$. Hence,
\[N_{\mathscr{E}}(p') = \left\{\mu \nabla c (p') + \rho : \mu \geq 0 \text{ and } \rho: \Theta \to \BR \right\}.\]
Therefore, there are $\mu,\mu' \geq 0$ and $\rho,\rho' : \Theta \to \BR$ such that
\[\pi y - \nabla c(p') = \mu \nabla c(p') + \rho \text{ and }\pi b - \nabla c(p') = \mu' \nabla c(p') + \rho'.\]
Combining these equations gives
\[b = \frac{1+\mu'}{1+\mu} y - \frac{(\mu'+1)\rho - (\mu+1)\rho'}{\pi(\mu+1)}.\]
To complete the argument we need to show that the coefficient on $b$ is at least equal to $\alpha'$. Suppose, on the contrary, that this coefficient is less than $\alpha'$. Then, from the definition of $\alpha'$, $c(p')<k$. But this implies $\mu=\mu'=0$ so that $(1+\mu')/(1+\mu)=1$, and this is impossible because, by its definition, $\alpha' \leq 1$.

\end{proof}

Note that if a first-best contract is feasible, it is also Pareto optimal. Notice too that if there is a first-best Pareto optimal contract yielding the agent utility $r$, then all Pareto optimal contracts yielding the agent utility at least $r$ are first-best since the agent's utility can be increased without affecting his behavior. 

From here, it is easy to see that the minimum agent utility achievable with a first-best Pareto optimal contract is obtained by the contract $\alpha'y - \beta$, where $\beta:\Theta \to \BR$ brings the minimum contract payment in each state to zero: $\beta(\theta) = \min\{\alpha 'y(d,\theta): d\in D\}$.
%

Also note that $\alpha'$ is increasing in the agent's capacity $k$.

\paragraph{Dealing with the capacity constraint}\label{CapacityConstraint}

In this section we show that imposing a capacity constraint is equivalent to scaling output. This will allow us to proceed in solving for the second-best Pareto optimal contracts with a scaled level of output in place of the capacity constraint.

The key idea is to consider a perturbation of our model, parametrized by $\alpha \in [0,1]$, in which output $y$ in the principal's objective function is replaced by $\alpha y$ (the feasible contracts, distributions, and profiles remain unchanged). Denote the set of Pareto optimal profiles to the perturbed problem by $\mathscr{P}(\alpha)$, and denote by $\mathscr{P}(\alpha,r)$ the subset of $\mathscr{P}(\alpha)$ in which the agent's utility is minimized subject to it being at least equal to $r$.

Consider a Pareto optimal profile $(b,p) \in \mathscr{P}(1,r)$ to the unperturbed problem. The following Proposition shows that all solutions $\mathscr{P}(1,r)$ of this unperturbed problem (a problem in which the capacity constraint might bind) may be obtained as solutions $\mathscr{P}(\alpha^*,E_{p}[b]-c(p))$ to the perturbed problem for the $\alpha^* \in [0,1]$ defined in the proposition and that, if the agent's cost function is continuous, then for at least one solution of this problem the capacity constraint does not bind.

\begin{prop}\label{capacitythm}
Let $(b,p) \in \mathscr{P}(1,r)$. Then there exists $\alpha^* \in [0,1]$, defined as
\[\alpha^* = \sup \left\{\alpha \in [0,1] : c(p')<k \text{ for all } (b',p') \in \mathscr{P}(\alpha, E_p[b]-c(p))\right\},\]
so that for each $\alpha \in [\alpha^*,1]$,
\[\mathscr{P}(\alpha,E_p[b]-c(p)) \supseteq \mathscr{P}(1,r),\]
and, conversely, if $(b_\alpha,p_\alpha)\in \mathscr{P}(\alpha,E_{p}[b]-c(p))$ is such that $c(p_\alpha) = k$, then
\[(b_\alpha,p_\alpha) \in \mathscr{P}(1,r).\]
%

If the agent's cost function is continuous, then there is a Pareto optimal profile in $\mathscr{P}(\alpha^*,E_p[b]-c(p))$ that solves the problem obtained by removing the capacity constraint from this perturbed problem.

\end{prop}

Note that $\alpha^*$, by its definition, is increasing in the agent's capacity.


\begin{proof}

By assumption $(b,p)\in \mathscr{P}(1,r)$. Let $\alpha \in [0,1)$ (the argument for $\alpha =1$ is trivial) and $(b_\alpha,p_\alpha) \in \mathscr{P}(\alpha,E_p[b]-c(p))$. We will show:
\begin{equation}\label{ineq1}
E_{p}[y] - E_{p_\alpha}[y] \geq E_{p}[b] - E_{p_\alpha}[b_\alpha] \geq \alpha(E_{p}[y] - E_{p_\alpha}[y]) \geq 0,
\end{equation}
\begin{equation}\label{ineq2}
c(p) - c(p_\alpha) \geq  E_{p}[b] - E_{p_\alpha}[b_\alpha].
\end{equation}
Then, if $\alpha \geq \alpha^*$ we may assume that $c(p_\alpha) = k$, and the inequalities in (\ref{ineq2}) give $c(p) = k$ and $E_p[b] = E_{p_\alpha}[b_\alpha]$, and so the inequalities in (\ref{ineq1}) give $E_p[y] = E_{p_\alpha}[y]$. Therefore $(b,p)\in \mathscr{P}(\alpha,E_p[b]-c(p))$ (so that $\mathscr{P}(\alpha,E_p[b]-c(p)) \supseteq \mathscr{P}(1,r)$) and $(b_\alpha,p_\alpha) \in \mathscr{P}(1,r)$.

To obtain the inequalities in (\ref{ineq1}) note that $E_p[y-b] \geq E_{p_\alpha}[y-b_\alpha]$ and $E_{p_\alpha}[\alpha y-b_\alpha] \geq E_p[\alpha y-b]$. The first and second inequalities of (\ref{ineq1}) are rearrangements of these; and from this the last inequality in (\ref{ineq1}) follows since $\alpha \in [0,1)$. Inequality (\ref{ineq2}) follows from the participation constraint $E_{p_\alpha}[b_\alpha] - c(p_\alpha) \geq E_{p}[b] - c(p)$. 

What remains to be shown is that there is a Pareto optimal profile in $\mathscr{P}(\alpha^*,E_p[b]-c(p))$ that solves the problem obtained by removing the capacity constraint from the perturbed problem for $\alpha^*$. Let $\alpha_1,\alpha_2,\cdots$ be a sequence in $[0,\alpha^*)$ converging to $\alpha^*$ and let $(b_{1},p_{1}), (b_{2},p_{2}),\cdots$ be a sequence of profiles with $(b_{i},p_{i}) \in \mathscr{P}(\alpha_{i},E_p[b]-c(p))$, for $i=1,2,\cdots$. Since the agent's cost function is continuous, the set of feasible profiles yielding the agent at least utility $E_p[b]-c(p)$ is compact. Therefore, there is a sub-sequence $(b_{n_1},p_{n_1}), (b_{n_2},p_{n_2}),\cdots$ converging to some feasible $(b',p')$. Continuity of the principal's objective function implies $(b',p') \in \mathscr{P}(\alpha^*,E_p[b]-c(p))$. Continuity of the agent's cost function together with the fact that $c(p_i)<k$ for $i=1,2,\cdots$ (because $\alpha_i<\alpha^*$) implies that $(b',p')$ solves the problem obtained by removing the capacity constraint from the perturbed problem for $\alpha^*$.
%
%
\end{proof}

\paragraph{Optimal distortion}

Having dealt with first-best Pareto optimal contracts (and developed a tool to jettison the capacity constraint) we can now turn to their second-best counterparts. Proposition~\ref{capacitythm} tells us that we can do away with the capacity constraint and instead replace output by $\alpha^*y$. In the next proposition we solve for the form of the optimal contract in the resulting problem.


\begin{prop}\label{prop3}
Let $(b,p)$ belong to $\mathscr{P}(\alpha,r)$ and assume that 
\begin{itemize}
\item the agent's capacity constraint does not bind: $p$ is an optimal experiment for the agent without a capacity constraint when the contract is $b$;
\item the agent's cost function $c$ is strictly convex and its second derivative is continuous;
\item $p$ assigns positive probability to each decision in each state.
\end{itemize}
Then 
\[b(d,\theta)=\alpha y(d,\theta)-\beta(\theta)-\gamma(d,\theta),\]
\[\gamma(d,\theta) = \frac{1}{\pi(\theta)}\sum_{ d',\theta' }\frac{\partial^2 c(p)}{\partial p(d|\theta)\partial p(d'|\theta')}\left(p(d'|\theta')(1-\xi) - \frac{\lambda[d',\theta']}{\pi(\theta')}\right).\addtag \label{gammaexp}\] 
The term $\lambda[d,\theta]$ is a Lagrange multiplier for the constraint $0\leq b(d,\theta)\leq y(d,\theta)$: it is non-negative when only the agent's liability limit binds, non-positive when only the principal's liability limit binds, zero when neither liability limit binds, and unrestricted in its value if both liability limits bind; $\xi\in [0,1]$ and is a Lagrange multiplier on the agent's participation constraint.
\end{prop}

By Proposition~\ref{capacitythm}, when the agent is capacity constrained there is a Pareto optimal contract of the form 
\[
b(d,\theta) = \alpha^* y(d,\theta) - \beta(\theta) - \gamma(d,\theta),
\]
where $\alpha^*\in [0,1]$ and is increasing in the agent's capacity. For second best contracts we have that 
\[
\beta(\theta)= \min\{\alpha^* y(d,\theta)-\gamma(d,\theta): d\in D\}.
\]

Note that if the mapping from $\alpha \in [0,1]$ to the set of Pareto optimal profiles $\mathscr{P}(\alpha,r)$ is lower hemi-continuous at $\alpha^*$, then all second-best Pareto optimal contracts take this form (cf. Proposition~\ref{capacitythm}).

\begin{proof}

\begin{center}
\emph{Step 1: necessary and sufficient conditions for optimal experiments}
\end{center}

The agent's problem is convex and so a necessary and sufficient condition for an experiment $p$ to be optimal against contract $b$ is that $\nabla_p \left(E_{p}[b]-c(p)\right)$ belong to the normal cone $N_\mathscr{E}(p)$ (see Theorem~\ref{wetrock6.12}):
\[
\pi(\theta)b(d,\theta) - \frac{\partial c(p)}{\partial p(d|\theta)} = \rho[\theta] \text{ for all }d\in D, \theta\in \Theta \addtag \label{agentcondition}
\]
for some $\rho \in \BR^{|\Theta|}$. Here we have used that the set of feasible experiments is defined without a capacity constraint and that the experiment $p$, by assumption, places positive probability on each decision in each state.

\begin{center}
\emph{Step 2: a necessary condition for optimal contracts} 
\end{center}

Condition~\ref{agentcondition} involves Lagrange multipliers $\rho$ so we shall pose the problem of designing an optimal contract as a choice of a contract $b$, experiment $p$, and Lagrange multipliers $\rho$ such that the contract and experiment form a feasible profile that yields the agent at least utility $r$. In the notation of Theorem~\ref{wetrock6.14} define $C = \{(b,p,\rho)\in A : F(b,p,\rho) \in B\},$ where $A = \BR^{|D\times\Theta|+|D\times \Theta|+|\Theta|}$, $F:\BR^{|D\times\Theta|+|D\times \Theta|+|\Theta|} \to \BR^{|D\times\Theta|+1+|\Theta|+|D\times \Theta|}$ is defined componentwise $F(b,p,\rho) = (f^i_{j}(b,p,\rho))_{i,j}$ by
\[
f^i_{j}(b,p,\rho) =\begin{cases}
b(d,\theta) & \text{ for } i=1, j\in D\times \Theta,\\
r- (E_p[b]-c(p)) & \text{ for } i=2,\\ 
\sum_{d\in D}p(d|\theta) - 1 & \text{ for } i=3,j\in \Theta,\\ 
\pi(\theta)b(d,\theta) - \frac{\partial c(p)}{\partial p(d|\theta)} - \rho[\theta]  & \text{ for } i=4, j \in D\times \Theta,
\end{cases}
\]
and $B=\left(\prod_{(d,\theta)\in D\times \Theta} [0,y(d,\theta)]\right) \times (-\infty,0] \times \{0\}^{|\Theta|+|D\times \Theta|}$.

The Pareto optimal contracts in which the principal maximizes his utility subject to the agent receiving at least utility $r$ is given by maximizing $E_p[y-b]$ over $C$. A necessary condition for the local optimality of $(b,p,\rho)$ is that $\nabla_{(b,p,\rho)} E_p[\alpha^* y-b]$ belongs to the normal cone $\hat N_C(b,p,\rho)$

The main challenge of the proof is in showing that $C$ is regular at any solution $(b,p,\rho)$ so that we may then express its normal cone as certain linear combinations of the gradients of the constraint functions.

We follow the notation and conditions of Theorem~\ref{wetrock6.14}. Thus, first note that by Theorem~\ref{wetrock6.15}, $A$ and $B$ are regular since they are products of intervals. Also, $N_A(b,p,\rho)=\{0\}$ and $N_B(b,p,\rho)$ is the standard multiplier cone consisting of Lagrange multipliers $(\lambda,\xi,\tau,\phi) \in \BR^{|D\times \Theta|}\times \BR \times \BR^{|\Theta|} \times \BR^{|D\times \Theta|}$ with each multiplier being zero if an interior point, non-positive if a right endpoint only, non-negative if a left end-point only, and unrestricted if both right and left end-points of the corresponding interval of $B$.

We now show that the constraint qualification of Theorem~\ref{wetrock6.14} holds at $(b,p,\rho)$. To do this consider the matrix obtained from the Jacobian of $F$ after eliminating the row corresponding to the gradient of the participation constraint function ($i=2$). This matrix has the form
\[
\begin{bmatrix}
I & 0 \\
\tilde C & \tilde D
\end{bmatrix}
\]
where $\tilde D$ is a square matrix and $I$ is an identity matrix. Therefore, its determinant, therefore, is given by the product of the determinant of the identity matrix and the determinant of $\tilde D$. The matrix $\tilde D$ has the form
\[
\begin{bmatrix}
\tilde B & 0 \\

\tilde A & -\tilde B'
\end{bmatrix}.
\]
where $\tilde A = \nabla^2 c(p)$ and $\tilde B = [I_{|\Theta|\times|\Theta|} \dots I_{|\Theta|\times|\Theta|}$]. It is enough to show that the matrix obtained from $\tilde D$ by interchanging the first block and the second block of rows is invertible. Block-wise inversion of this matrix gives
\[
 \begin{bmatrix} (\tilde B\tilde A^{-1}\tilde B')^{-1} & -(\tilde B\tilde A^{-1}\tilde B')^{-1}\tilde B\tilde A^{-1} \\ -\tilde A^{-1}(-\tilde B')(\tilde B\tilde A^{-1}\tilde B')^{-1} & \quad \tilde A^{-1}+\tilde A^{-1}(-\tilde B')(\tilde B\tilde A^{-1}\tilde B')^{-1}\tilde B\tilde A^{-1}\end{bmatrix}.
\]
Hence it is invertible if $\tilde A$ and $\tilde B\tilde A^{-1}\tilde B'$ are invertible. $\tilde A$ is invertible because it is positive definite as the Hessian of a strictly convex function. To see that $\tilde B\tilde A^{-1}\tilde B'$ is invertible, let $x$ be any nonzero vector conformable for pre-multiplication with $\tilde B$. Then, since $\tilde A^{-1}$ is positive definite, 
\[
0< x' \tilde B \tilde A^{-1} \tilde B'x,
\]
which implies that $\tilde B\tilde A^{-1}\tilde B'$ is positive definite and thus invertible.
We can conclude from this that the set of gradients of the functions $f^i_j$ ($i\neq 2$) at $(b,p,\rho)$ form a linearly independent set of vectors. 

At least one of the liability limits must not bind, and so we can replace the corresponding gradient with the gradient of $f^2_j$. Since $p$ is assumed to place positive probability on each decision in each state and the prior has full support, the first $|D\times \Theta|$ components of this gradient are nonzero. Thus the new set of vectors form a linearly independent set. Therefore the standard constraint qualification of Theorem~\ref{wetrock6.14} holds at $(b,p,\rho)$ and hence $C$ is regular at $(b,p,\rho)$ (so that $\hat N_C(b,p,\rho)=N_C(b,p,\rho))$ and hence
\begin{equation*} 
     \begin{split}
     N_C(b,p,\rho) =
     \left\{\sum_{d,\theta}\lambda[d,\theta] \nabla f^1_{(d,\theta)}(b,p,\rho) + \xi \nabla f^2(b,p,\rho) + \sum_{\theta }\tau[\theta] \nabla f^3_{\theta}(b,p,\rho) \right. \\ \left.  + \sum_{d,\theta}\phi[d,\theta] \nabla f^4_{(d,\theta)}(b,p,\rho): (\lambda,\xi,\tau,\phi) \in N_B(F(b,p,\rho))\right\}.
    \end{split}
\end{equation*}
The condition $\nabla_{(b,p,\rho)} E_p[y-b]  \in \hat N_C(b,p,\rho)$ may now be written
\[
\pi(\theta) p(d|\theta) = \lambda[d,\theta] + \xi \pi(\theta) p(d|\theta) + \phi[d,\theta] \pi(\theta),
\]
\[
\pi(\theta)(y(d,\theta)-b(d,\theta)) = \tau[\theta]
\]
\[
 +  \sum_{d',\theta'}\phi[d',\theta']\frac{\partial^2 c(p)}{\partial p(d|\theta)\partial p(d'|\theta')} + \xi \left(\pi(\theta)b(d,\theta) - \frac{\partial c(p)}{\partial p(d|\theta)}\right),
\]
\[
0 = -\sum_{d'}\phi[d',\theta].
\]
Therefore
%
\[
b(d,\theta)=y(d,\theta)-\frac{\tau[\theta]+\xi\rho[\theta]}{\pi(\theta)}-\frac{1}{\pi(\theta)}\sum_{d',\theta'}\phi[d',\theta']\frac{\partial^2 c(p)}{\partial p(d|\theta)\partial p(d'|\theta')}. \addtag \label{optimalcontractform}
\]
where $\phi[d,\theta] = p(d|\theta)(1-\xi)-\lambda[d,\theta]/\pi(\theta)$. Defining $\beta:\Theta\to \BR$ and $\gamma:D\times \Theta \to \BR$ as 
\[
\beta(\theta) = \frac{\tau[\theta]+\xi\rho[\theta]}{\pi(\theta)},
\]
\[
\gamma(d,\theta) = \frac{1}{\pi(\theta)}\sum_{ d',\theta' }\left(p(d'|\theta')(1-\xi) - \frac{\lambda[d',\theta']}{\pi(\theta')}\right)\frac{\partial^2 c(p)}{\partial p(d|\theta)\partial p(d'|\theta')},
\]
%
%
%
completes the proof.
\end{proof}

Together, propositions~\ref{prop2},~\ref{capacitythm}, and~\ref{prop3} provide a general characterization of Pareto optimal contracts in terms of a piece rate $\alpha\in [0,1]$, a state-dependent transfer $\beta:\Theta \to \BR$, and an optimal distortion $\gamma:D\times \Theta \to \BR$.

For a first-best contract, $\gamma=0$ and there is no distortion to incentives. The piece rate $\alpha$ may be any value in the interval $[\alpha',1]$, $\alpha'$ being the point below which the agent's capacity is no longer exhausted. The state-dependent transfer $\beta$ is any state dependent transfer such that the resulting contract $\alpha y - \beta$ satisfies the liability limits and so is feasible.

For a second-best contract the piece rate is equal to $\alpha^*$ and the state-dependent transfer $\beta$ is given by $\beta(\theta) = \min\{\alpha^* y(d,\theta)-\gamma(d,\theta): d\in D\}$ so that liability limits bind in each state. The non-distortionary ways of transferring utility from the agent to the principal are fully exploited. 

The term $\gamma$ describes the optimal distortion to incentives---that is, it is the way of transferring utility from the agent to the principal at minimum loss to welfare (expected output less experiment cost). 

Expression~\ref{gammaexp} sheds light on how this optimal distortion is constructed showing that it largely depends on the complementarities in the cost of acquiring different signals in different states. To see this, suppose that the agent's liability limit binds in state $\theta'$ following decision $d'$ so that the Lagrange multiplier $\lambda[d',\theta']>0$ (i.e. the optimal contract would be different without this liability limit). To lever down the probability of receiving the signal to take decision $d'$ when the state is $\theta'$ the principal can no longer reduce the contract payment $b(d',\theta')$. As a result, other contract payments---those not immobilized against their liability limits---are recruited for this task. Expression~\ref{gammaexp} shows the way this happens. If a contract payment $b(d,\theta)$ is not stuck against the agent's liability limit then complementarity in the cost of signal probabilities $p(d'|\theta')$ and $p(d|\theta)$ together with $\lambda[d',\theta']>0$ adds to $\gamma(d,\theta)$. This reduces the contract payment $b(d,\theta)$ which in turn reduces the signal probability $p(d|\theta)$ and hence signal probability $p(d'|\theta')$. 

The term $p(d'|\theta')(1-\xi)$ in Expression~\ref{gammaexp} accounts for adjustments in the contract payments $b(d,\theta)$ necessary to keep the incentive and participation constraints holding following a reduction in $p(d'|\theta')$.   


\section{Optimal distortion under different costs}\label{poscomp}

\paragraph{Mutual Information}\label{simpledescription}

A simple and intuitive form for $\gamma$ arises when the agent's information cost matrix $k$ is given by the inverse Fisher information matrix
\[
k(\theta,\theta', p(\cdot|d)) = 
\begin{cases}
p(\theta|d)(1-p(\theta|d)) & \text{ if } \theta = \theta'\\
-p(\theta|d)p(\theta'|d) & \text{ if } \theta \neq \theta'
\end{cases}.
\]
Then the agent's cost function is expected reduction in Shannon entropy
\[c(p) = H_S(\pi) - \sum_{d\in D}p(d)H_S(p(\cdot|d)),\]
where $H_S$ denotes Shannon entropy, defined for each $q\in \Delta(\Theta)$ by
\[H_S(q) = -\sum_{\theta\in \Theta} q(\theta)\log q(\theta).\]
Then 
\begin{align*}
\gamma(d,\theta) = \sum_{\bar \theta}\frac{\lambda[d,\bar\theta]}{p(d)}-\frac{\lambda[d,\theta]}{p(d|\theta)\pi(\theta)}
\end{align*}
and so the optimal contract has the following simple description:

\begin{itemize}
\item \textbf{Decision $d$ penalty/reward:} 
$\hat \gamma(d) = \sum_{\bar \theta}\frac{\lambda[d,\bar\theta]}{p(d)};$
\item \textbf{Indexed output:} $y_I = y - \beta/\alpha^*$ where 
$\beta(\theta) = \min\{\alpha^* y(d,\theta)-\hat \gamma(d):d\in D\};$
\item \textbf{Piece rate:} $\alpha^* \in [0,1]$.
\end{itemize}
Here we assume the liability limits are implicitly understood: if the sum of promised payments falls outside the allowed range, the actual payment is the closest liability limit.

Notice that the decision $d$ penalty/reward depends on the number of states in which liability limits bind following decision $d$. When the agent's liability limit binds $\lambda[d,\theta] > 0$ which pushes the decision $d$ transfer in the direction of a punishment; when the principal's liability limit binds $\lambda[d,\theta] < 0$ which pushes the decision $d$ transfer in the direction of a reward. The magnitude of the multiplier indicates the cost to the principal of having the liability limit bind.

One can see from Expression (\ref{genlD}) that we recover the same contract form as when the agent's cost function is expected reduction in Shannon entropy provided the information cost matrix is of the form 
\[
k(\theta,\theta', p(\cdot|d))= p(\theta|d)g(\theta',p(\cdot|d))+ \mathbf{1}_{\{\theta'=\theta\}}h(\theta,p(\cdot|d))
\]
for some functions $f$ and $g$. However, from this form it is straightforward to show (using symmetry of the information cost matrix and that its rows sum to zero) that the cost function is necessarily proportional to expected reduction in Shannon entropy.

\paragraph{Expectation of Bregman divergence between prior and posterior}

Here we consider a class of cost functions defined in~\cite{HebertWoodford16}. These cost functions emerge from a dynamic information acquisition problem in which an agent at each point in time decides whether to take a decision or continue collecting information. Information collection is modeled as the agent choosing the covariances of a multi-dimensional Brownian motion that moves about the set of probability distributions over states, subject to a constraint on these covariances in terms of the complementarities and substitutabilities of acquiring information about different states. They show that this dynamic problem has a static representation in which the agent's cost function is defined as
\[
c(p) = \sum_{d\in D}p(d)D(p(\cdot|d)||\pi)
\]
where
\[
D(p(\cdot|d)||\pi) = H(p(\cdot|d)) - H(\pi) - (p(\cdot|d)-\pi)^T \nabla H(\pi)
\]
is the Bregman divergence associated with the convex function $H$ and where
\[
\frac{\partial}{\partial p(\theta|d)}\frac{\partial}{\partial p(\theta'|d)}H(p(\cdot|d)) = \frac{k(\theta,\theta', p(\cdot|d))}{p(\theta|d)p(\theta'|d)}. \addtag \label{genlD}
\]
The matrix $k$ is called the information cost matrix and characterizes the complementarities and substitutabilities in acquiring information about different states. It is positive semi-definite, symmetric, and its rows sum to zero. Given experiment $p$, there is complementarity in learning about distinct states $\theta$ and $\theta'$ when $k(\theta,\theta',p(\cdot|d))$ is negative and substitutabilities when it is positive; when $\theta = \theta'$ it is nonnegative and measures the difficulty of learning about state $\theta$.

For any cost function in this class, its Hessian is given by (see section~\ref{BW} for the derivation)
\[
\frac{\partial^2 c(p)}{\partial p(d|\theta)\partial p(d|\theta')}
 = p(d)\frac{k(\theta,\theta',p(\cdot|d))}{p(d|\theta)p(d|\theta')},
\]
and $0$ otherwise. Substitution into (\ref{gammaexp}), the expression for $\gamma$, gives
\[
\gamma(d,\theta)= -\frac{1}{\pi(\theta)}\sum_{\theta'}\frac{\partial^2 c(p)}{\partial p(d|\theta)\partial p(d|\theta')}\frac{\lambda[d,\theta']}{\pi(\theta')}.
\]

\section{Graphical method}

Here we use geometric methods to study the optimal contract when the agent's cost function is posterior separable. The agent's cost function is posterior-separable if there exists a nonnegative real valued function $\Upsilon$ on $\Delta(\Theta)$ such that 
\[
c(p) = \Upsilon(\pi) - \sum_{d\in D}p(d)\Upsilon(p(\cdot|d)).
\]

 We will refer to $\Upsilon$ the \emph{uncertainty function}. Note that $\Upsilon$ is concave if and only if $c$ takes on nonnegative values for all prior distributions and all experiments (\cite{degroot1962uncertainty}).\footnote{One attractive feature of posterior separable cost functions is that if one thinks of an experiment as a lottery over posterior distributions, then a cost function that is posterior separable is analogous to an expected utility function with the uncertainty function in place of Bernoulli utility.} Since we have assumed that the agent's cost functions takes on nonnegative values, $\Upsilon$ will always be a concave function.\footnote{My representation of the agent's preferences seems only to ensure that the cost function is nonnegative when his prior is $\pi$. I am implicitly assuming that the agent's preferences are stable in the sense that if he receives new information before interacting with the principal (and thus obtains a new prior), then his preferences are given by the same utility function but with the new prior in place of the old.}

 A mixture of the experiments $p$ and $p'$ with weighting $t$ and $1-t$ corresponds to a mixture of their posterior distributions on $\Theta$ with weighting $tp(d)/(tp(d)+(1-t)p'(d))$ and $(1-t)p'(d)/(tp(d)+(1-t)p'(d)).$ Basic calculation then shows that the agent's cost function $c$ is a convex function on the set of experiments. This proves the following lemma.

 \begin{lma}\label{convexcost}
 If the agent's cost function is posterior separable, then it is convex.
 \end{lma}

%
%
%
%

A benefit of using posterior separable cost functions is that they allow for a simple geometric representation of the agent's problem.\footnote{This representation is taken from~\cite{caplin2013behavioral}; the Bayesian persuasion literature initiated by~\cite{kamenica2011bayesian} develops a geometric approach similar to the one here using insights from~\cite{aumann1995repeated}.} Figure~\ref{geosol1} represents the agent's problem of choosing an experiment under a posterior separable cost function when there are two states and two decisions.

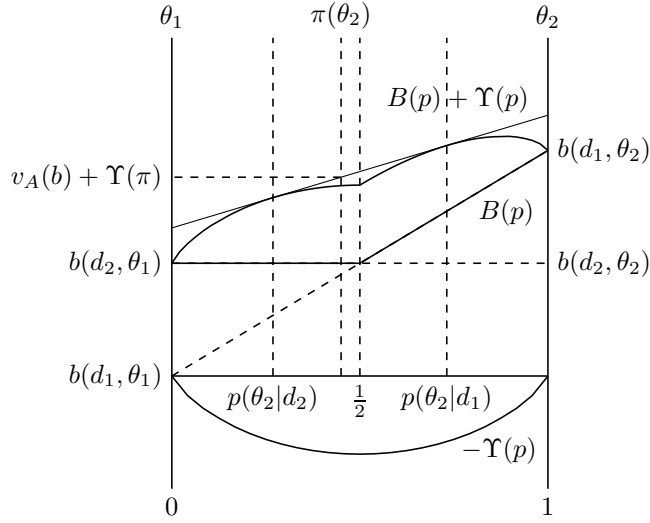
\begin{figure}[h]
\centering
\begin{tikzpicture}[xscale=5,yscale=1.5]
 \draw[semithick,-] (0,0) -- (1,0); 
\draw [semithick, -] (0,-1) node[below]{$0$} -- (0,3) node[above]{$\theta_1$}; 
\draw [semithick, -] (1,-1) node[below]{$1$}  -- (1,3) node[above]{$\theta_2$}; 
\draw [semithick, domain=0.0001:0.9999] plot (\x, {((\x)*ln(\x)+(1-\x)*ln(1-\x))})node at(0.87,-0.66){$-\Upsilon(p)$}; 
\draw [dashed, semithick,-] (0,0) node[left] {$b(d_1,\theta_1)$} -- (1,2) node[right] {$b(d_1,\theta_2)$} ;
\draw [dashed, semithick,-] (0,1)node[left] {$b(d_2,\theta_1)$} -- (1,1) node[right] {$b(d_2,\theta_2)$};
\draw [semithick, domain=0:0.5] plot (\x, {1});
\draw [semithick, domain=0.5:1] plot (\x, {2*\x}) node at(0.89,1.45){$B(p)$};
\draw [semithick, domain=0.0001:0.5]  plot (\x, {-((\x)*ln(\x)+(1-\x)*ln(1-\x))+1}) ;
\draw [semithick, domain=0.5:0.9999]  plot (\x, {-((\x)*ln(\x)+(1-\x)*ln(1-\x))+2*\x}) node at(0.76,2.45){$B(p)+\Upsilon(p)$};
\draw [domain=0:1]  plot (\x, {\x+1.31326});
\draw [semithick, dashed ,-] (0.268941,0) node[below] {$p(\theta_2|d_2)$} -- (0.268941,3);
\draw [semithick, dashed ,-] (0.731059,0) node[below] {$p(\theta_2|d_1)$} -- (0.731059,3);
\draw [semithick, dashed ,-] (0.45,0) -- (0.45,3) node[above] {$\pi(\theta_2)$};
\draw [semithick, dashed ,-] (0.5,0) node[below] {$\frac{1}{2}$} -- (0.5,3) ;
\draw [semithick, dashed ,-] (0.45,1.76326) -- (0,1.76326) node[left] {$v_A(b)+\Upsilon(\pi)$};
\end{tikzpicture}
\caption{\small{The agent's problem for $\pi(\theta_2)=0.45$ and contract $b(d_1,\theta_1)=0$, $b(d_2,\theta_1)=b(d_2,\theta_2)=1$, and $b(d_1,\theta_2)=2$. Reduced form $B(p)$ and upper envelope of net utilities $B(p)+\Upsilon(p)$ as a function of the posterior probability assigned to state $\theta_2$. Optimal posteriors are given by maximizing expected net utility subject to the average posterior equaling the prior.}}\label{geosol1}
\end{figure}
Contract payments in states $\theta_1$ and $\theta_2$ are shown on the left and right vertical axes. The horizontal axis measures the probability assigned to state $\theta_2$. The dashed lines connecting contract payments in different states show the expected payoffs of each decision as a function of the probability assigned to state $\theta_2$. The upper envelope of these lines is called the \emph{reduced form} of contract $b$ and is denoted by $B(p)$ (That is, $B:\Delta(\Theta)\to \BR$ is defined as the agent's maximal expected contract payment as a function of his posterior belief about the state: $B(p) = \max\{E_p[b(d,\cdot)]:d\in D\}$). The reduced form describes the decision the agent will take as a function of his posterior belief about the state.
For the contract shown, the agent will choose decision $d_2$ when the probability he assigns to state $\theta_2$ is less than $1/2$ and he will choose decision $d_1$ when the probability he assigns to state $\theta_2$ is greater than $1/2$. 

Because the cost function is posterior separable the agent's objective function may be written as 
\begin{align*}
\sum_{d}p(d)\left(\sum_{\theta}p(\theta|d)b(d,\theta) + \Upsilon(p(\cdot|d))\right) - \Upsilon(\pi).
\end{align*} 
Following~\cite{caplin2013behavioral}, the term $\sum_{\theta}p(\theta|d)b(d,\theta) + \Upsilon(p(\cdot|d))$ is called the \emph{net utility} of decision $d$. The upper envelope of the net utility for decision $d_1$ and the net utility for decision $d_2$ is given by $B(p)+\Upsilon(p)$. Note that the average of the agent's posterior beliefs is equal to his prior belief and therefore the agent's utility is given by weighting the optimal net utilities---given by $B+\Upsilon$---with weights that average the posteriors to the prior. In other words, the optimal posteriors are computed by considering the ``best'' tangent hyperplane to the upper envelope of the net utility function, $B+\Upsilon$, or, alternatively, as in~\cite{kamenica2011bayesian}, by evaluating the concavification of $B+\Upsilon$ at the prior.\footnote{The \emph{concavification} of a function is the smallest concave function everywhere weakly greater than the function.} In the figure, the optimal posterior probabilities are shown as well as the point $(\pi(\theta_2),v_A(b)+\Upsilon(\pi))$ corresponding (but not equal) to the agent's maximized utility.

\paragraph{Altering a contract by a state-dependent transfer.} 

The agent's behavior is unchanged when his contract is altered by a state-dependent transfer since $E_p[b-\beta]-c(p) = E_p[b] - E_\pi[\beta]-c(p)$ for all $\beta:\Theta \to \BR$. Figure~\ref{geosol2} illustrates this for the contract shown in black and the state contingent payment $\beta(\theta_1)=0$, $\beta(\theta_2)=1$. Note that the optimal experiment would not change if $\beta(\theta_2)>1$ but the resulting contract would not be feasible.
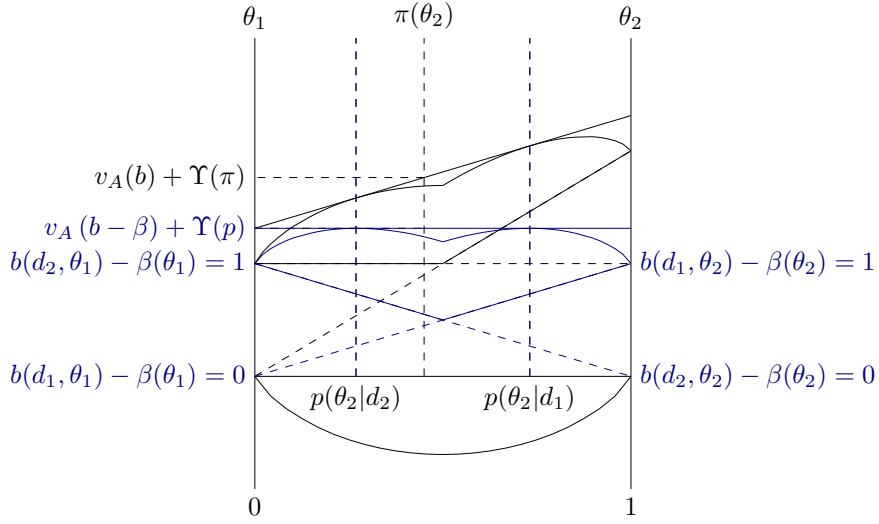
\begin{figure}[H]
\centering
\begin{tikzpicture}[xscale=5,yscale=1.5]
 \draw[very thin,-] (0,0) -- (1,0); 
\draw [very thin, -] (0,-1) node[below]{$0$} -- (0,3) node[above]{$\theta_1$}; 
\draw [very thin, -] (1,-1) node[below]{$1$}  -- (1,3) node[above]{$\theta_2$}; 
\draw [very thin, domain=0.0001:0.9999] plot (\x, {((\x)*ln(\x)+(1-\x)*ln(1-\x))}); 
\draw [dashed, very thin,-] (0,0) -- (1,2);
\draw [dashed, very thin,-] (0,1) -- (1,1);
\draw [very thin,domain=0:0.5] plot (\x, {1});
\draw [very thin,domain=0.5:1] plot (\x, {2*\x});
\draw [very thin, domain=0.0001:0.5]  plot (\x, {-((\x)*ln(\x)+(1-\x)*ln(1-\x))+1});
\draw [very thin, domain=0.5:0.9999]  plot (\x, {-((\x)*ln(\x)+(1-\x)*ln(1-\x))+2*\x});
\draw [very thin, domain=0:1]  plot (\x, {\x+1.31326});
\draw [very thin, dashed ,-] (0.268941,0) node[below] {$p(\theta_2|d_2)$} -- (0.268941,3);
\draw [very thin, dashed ,-] (0.731059,0) node[below] {$p(\theta_2|d_1)$} -- (0.731059,3);
\draw [very thin, dashed ,-] (0.45,0) -- (0.45,3) node[above] {$\pi(\theta_2)$};
\draw [very thin, dashed ,-] (0.45,1.76326) -- (0,1.76326) node[left] {$v_A(b)+\Upsilon(\pi)$};
\draw [NavyBlue, dashed, very thin,-] (0,0) node[left] {$b(d_1,\theta_1)-\beta(\theta_1)=0$} -- (1,1) node[right] {$b(d_1,\theta_2)-\beta(\theta_2)=1$} ;
\draw [NavyBlue, dashed, very thin,-] (0,1) node[left] {$b(d_2,\theta_1)-\beta(\theta_1)=1$} -- (1,0) node[right] {$b(d_2,\theta_2)-\beta(\theta_2)=0$};
\draw [NavyBlue, very thin,domain=0:0.5] plot (\x, {1-\x});
\draw [NavyBlue, very thin,domain=0.5:1] plot (\x, {\x});
\draw [NavyBlue, very thin, domain=0.0001:0.5]  plot (\x, {-((\x)*ln(\x)+(1-\x)*ln(1-\x))+1-\x});
\draw [NavyBlue, very thin, domain=.5:0.9999]  plot (\x, {-((\x)*ln(\x)+(1-\x)*ln(1-\x))+\x});
\draw [NavyBlue, very thin, domain=0:1]  plot (\x, {1.31326});
\draw [NavyBlue, very thin, dashed ,-] (0.268941,0) -- (0.268941,3);
\draw [NavyBlue, very thin, dashed ,-] (0.731059,0) -- (0.731059,3);

\draw [NavyBlue, very thin, dashed ,-] (0,1.31326) node[left] {$v_A\left(b-\beta\right)+ \Upsilon(p)$} -- (0.45,1.31326);
\end{tikzpicture}
\caption{\small Altering the contract $b$ by the state contingent payment $\beta(\theta_1) = 0$, $\beta(\theta_2)=1$ does not alter the experiment the agent chooses.}\label{geosol2}
\end{figure}

\paragraph{Scaling a contract when capacity constraints binds.} 

When the agent's capacity constraint binds under his chosen experiment it is possible to scale down---multiply by a factor less than one---the contract without altering the agent's behavior. Figure~\ref{cost} shows that imposing a binding capacity constraint alters the agent's chosen experiment by bringing the posterior probabilities closer to the prior. Figure~\ref{scale} shows that given the binding capacity constraint the contract $b$ can be scaled down by any factor $\alpha$ in $[1/2,1]$ without altering the agent's chosen experiment.

\begin{figure}[H]
\centering
\subfloat[]{
\begin{tikzpicture}[xscale=5,yscale=1.5]
 \draw[very thin,-] (0,0) -- (1,0); 
\draw [very thin, -] (0,-1) node[below]{$0$} -- (0,3) node[above]{$\theta_1$}; 
\draw [very thin, -] (1,-1) node[below]{$1$}  -- (1,3) node[above]{$\theta_2$}; 
\draw [very thin, domain=0.0001:0.9999] plot (\x, {((\x)*ln(\x)+(1-\x)*ln(1-\x))}); 
\draw [dashed, very thin,-] (0,0) node[left] {$b(d_1,\theta_1)=0$} -- (1,2) node[right] {$b(d_1,\theta_2)=2$} ;
\draw [dashed, very thin,-] (0,1) node[left] {$b(d_2,\theta_1)=1$} -- (1,1) node[right] {$b(d_2,\theta_2)=1$};
\draw [very thin,domain=0:0.5] plot (\x, {1});
\draw [very thin,domain=0.5:1] plot (\x, {2*\x});
\draw [very thin, domain=0.0001:0.5]  plot (\x, {-((\x)*ln(\x)+(1-\x)*ln(1-\x))+1});
\draw [very thin, domain=0.5:0.9999]  plot (\x, {-((\x)*ln(\x)+(1-\x)*ln(1-\x))+2*\x});
\draw [very thin, domain=0:1]  plot (\x, {\x+1.31326});
\draw [very thin, dashed ,-] (0.268941,0) -- (0.268941,3);
\draw [very thin, dashed ,-] (0.731059,0) -- (0.731059,3);
\draw [very thin, dashed ,-] (0.45,0) -- (0.45,3) node[above] {$\pi(\theta_2)$};
\draw [very thin, dashed ,-] (0.45,1.76326) -- (0,1.76326);
\draw [BrickRed, very thin, domain=0:1]  plot (\x, {\x+1.28531});
\draw [BrickRed, very thick, dashed ,-] (0.377541,-0.662847) -- (0.377541,3);
\draw [BrickRed, very thick, dashed ,-] (0.622459,-0.662847) -- (0.622459,3);
\draw [BrickRed, very thin, dashed ,-] (0,1.73531) -- (0.45,1.73531);
\draw [BrickRed, very thick, dashed ,-] (0,-0.662847)  -- (1,-0.662847);
\draw [BrickRed, very thick, dashed ,-] (0.45,-0.8) -- (0.45,-0.8) ;
\end{tikzpicture}\label{cost}}\\
\subfloat[]{
\begin{tikzpicture}[xscale=5,yscale=1.5]
 \draw[very thin,-] (0,0) -- (1,0); {}
\draw [very thin, -] (0,-1) node[below]{$0$} -- (0,3) node[above]{$\theta_1$}; 
\draw [very thin, -] (1,-1) node[below]{$1$}  -- (1,3) node[above]{$\theta_2$}; 
\draw [very thin, domain=0.0001:0.9999] plot (\x, {((\x)*ln(\x)+(1-\x)*ln(1-\x))}); 
\draw [dashed, very thin,-] (0,0) -- (1,2)  ;
\draw [dashed, very thin,-] (0,1)  -- (1,1) ;
\draw [very thin,domain=0:0.5] plot (\x, {1});
\draw [very thin,domain=0.5:1] plot (\x, {2*\x});
\draw [very thin, domain=0.0001:0.5]  plot (\x, {-((\x)*ln(\x)+(1-\x)*ln(1-\x))+1});
\draw [very thin, domain=0.5:0.9999]  plot (\x, {-((\x)*ln(\x)+(1-\x)*ln(1-\x))+2*\x});
\draw [very thin, domain=0:1]  plot (\x, {\x+1.31326});
\draw [very thin, dashed ,-] (0.45,0) -- (0.45,3) node[above] {$\pi(\theta_2)$};
\draw [BrickRed, very thin, domain=0:1]  plot (\x, {\x+1.28531});
\draw [BrickRed, very thin, dashed ,-] (0,1.73531) -- (0.45,1.73531);
\draw [BrickRed, very thick, dashed ,-] (0,-0.662847) node[left] {$\Upsilon(\pi)-k$} -- (1,-0.662847);
\draw [BrickRed, very thick, dashed ,-] (0.45,-0.8) -- (0.45,-0.8) ;
\draw [NavyBlue, dashed, very thin,-] (0,0) node[left] {$\frac{b(d_1,\theta_1)}{2}=0$} -- (1,1) node[right] {$\frac{b(d_1,\theta_2)}{2}=1$} ;
\draw [NavyBlue, dashed, very thin,-] (0,.5) node[left] {$\frac{b(d_2,\theta_1)}{2}=0.5$} -- (1,.5) node[right] {$\frac{b(d_2,\theta_2)}{2}=0.5$};
\draw [NavyBlue, very thick,domain=0:0.5] plot (\x, {1/2});
\draw [NavyBlue, very thick,domain=0.5:1] plot (\x, {\x});
\draw [NavyBlue, very thick, domain=0.0001:0.5]  plot (\x, {-((\x)*ln(\x)+(1-\x)*ln(1-\x))+0.5});
\draw [NavyBlue, very thick, domain=.5:0.9999]  plot (\x, {-((\x)*ln(\x)+(1-\x)*ln(1-\x))+\x});
\draw [NavyBlue, very thick, domain=0:1]  plot (\x, {0.5*\x+0.974077});
\draw [BrickRed, very thick, dashed ,-] (0.377541,-0.662847) -- (0.377541,3);
\draw [BrickRed, very thick, dashed ,-] (0.622459,-0.662847) -- (0.622459,3);

\draw [NavyBlue, very thin, dashed ,-] (0,1.19908) node[left] {$v_A\left(\frac{b}{2}\right)+\Upsilon(\pi)$} -- (0.45,1.19908);
\end{tikzpicture}\label{scale} }

\caption{\small \protect\subref{cost} A binding capacity constraint (thick dashed lines) brings the posterior probabilities closer to the prior. 
%
%
\protect\subref{scale} With a binding capacity constraint (thick dashed lines), scaling incentives in half (thick solid lines) does not alter the agent's chosen experiment. 
%
}
\end{figure}
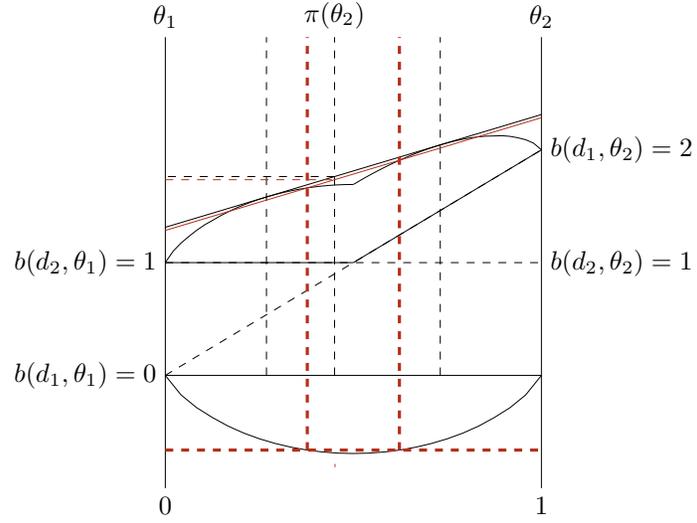
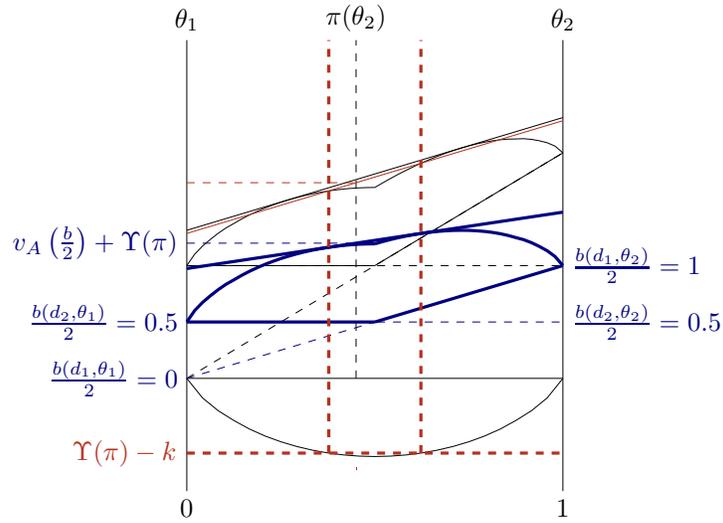

\section{Two state, two decision example}

Consider the following problem in which a principal hires an agent to choose between a safe and risky course of action:
\begin{figure}[h]
\centering
\begin{tabular}{ r !{\vrule width0.9pt}   c  c }
 $y(d,\theta)$ &  $\theta_1$ &  $\theta_2$  \\ \Xhline{0.9pt}
 $d_1$  & 0 & 10   \\
 $d_2$  & 5 & 5  
\end{tabular}
\end{figure}

Suppose the agent's prior belief on state $\theta_1$ is $2/3$ and that the set of feasible experiments are those with cost less than $k=1/2$.

\begin{center}
\emph{First-best contracts}
\end{center}

We first solve for Pareto optimal contracts that induce first-best experimentation. 
Since altering a contract by a state-dependent transfer does not alter the experiment chosen by the agent, any contract $y-\beta$ with $\beta:\Theta \to \BR$ induces first best experimentation. Similarly, if the agent's capacity constraint binds under these contracts, then there exists a number $\alpha'$ in $[0,1]$ (as defined in Proposition~\ref{prop2}) such that for each of the contracts $\alpha(y-\beta)$ with $\alpha$ in $[\alpha',1]$ the agent chooses the same first-best experiment. Therefore, each contract of the form
\begin{align}\label{first_best_set_of_contracts}
\alpha(y-\beta) \text{ with } \alpha \in [\alpha',1] \text{ and } \beta:\Theta\to \BR \text{ chosen so the contract is feasible}
\end{align}
is Pareto optimal and induces first-best experimentation. The least favorable of these contracts for the agent is $\alpha'(y+\beta')$ with $\beta'$ defined by $\beta'(\theta_1)=0$, $\beta'(\theta_2) = -5$. So the set of agent reservation utilities corresponding to these contracts is $[v_A(\alpha'(y-\beta')),v_A(y)]$.

Results from \cite{matejka2014rational} and \cite{caplin2013behavioral} show that when the agent's cost function is expected Shannon entropy reduction the optimal experiment resembles a logit-rule with the observable part of the utility from taking decision $d$ in state $\theta$ given by $b(d,\theta)/(1+\mu) + \log p(d)$:
\begin{align}\label{Shannonexperiment}
p(d|\theta) = \frac{p(d)e^\frac{b(d,\theta)}{1+\mu}}{\sum_{\bar d}p(\bar d)e^\frac{b(\bar d,\theta)}{1+\mu}}.
\end{align}
The term $\mu$ is a nonnegative dual variable which is chosen as small as possible subject to the agent's capacity constraint binding.
Using Equation (\ref{Shannonexperiment}), if the agent's capacity constraint were not to bind then the first-best experiment would be as shown in Table~\ref{first_best_experiment_unconstrained}.
\begin{table}[h]
\centering
\subfloat[]{
\begin{tabular}{ r !{\vrule width0.9pt} c  c }
 $p(d|\theta)$ &  $\theta_1$ &  $\theta_2$  \\ \Xhline{0.9pt}
 $d_1$  & 0.007   & 0.993   \\
 $d_2$  &  0.993 &  0.007
\end{tabular}\label{first_best_experiment_unconstrained}
}
\hspace{5em}
\subfloat[]{
\begin{tabular}{ r !{\vrule width0.9pt} c  c }
 $p(d|\theta)$ &  $\theta_1$ &  $\theta_2$  \\ \Xhline{0.9pt}
 $d_1$  & 0.031 &  0.969   \\
 $d_2$  & 0.969 &   0.031 
\end{tabular}\label{first_best_experiment_constrained}
}\captionof{table}{\small \protect\subref{first_best_experiment_unconstrained} First best experiment when the agent's capacity constraint does not bind. \protect\subref{first_best_experiment_constrained} First best experiment under capacity constraint $c(p)\leq 1/2$.}\label{Table1}
\end{table}
\noindent However, this experiment is not feasible since its cost is $0.596$. Instead, the agent chooses the experiment show in in Table~\ref{first_best_experiment_constrained} which has cost $1/2$ and is given by setting $\mu = 0.446$ in Equation (\ref{Shannonexperiment}) with $b=y$. This implies that $\alpha' = 1/(1+\mu) = 0.692$. Under this experiment, the agent is slightly more prone to error than under the unconstrained experiment. 

The minimum agent utility for a first-best contract is $v_A(\alpha'(y-\beta')) = 2.853$ and the maximum utility is $v_A(y) = 6.014$; the first-best Pareto optimal contracts given by Equation (\ref{first_best_set_of_contracts}) deliver the agent utilities in this range. Under each of these contracts, the agent chooses the experiment shown in Table~\ref{first_best_experiment_constrained}. Note that any Pareto optimal contract that delivers the agent a utility in this range must induce first best experimentation (otherwise the contracts found here would be an improvement, contradicting Pareto optimality).

\begin{center}
\emph{Second-best contracts}
\end{center}

In the previous section, we showed that if the agent's reservation utility is above $2.853$, then all Pareto optimal contracts induce first-best experimentation. We now consider the case where the agent's reservation utility is below this level. 

For a second-best contract the minimum payment in each state is $0$. Consider the case where $b(d_1,\theta_1)$ and $b(d_2,\theta_2)$ are $0$ and $b(d_1,\theta_2)$ and $b(d_2,\theta_1)$ are positive and such that the principal's liability limits do not bind. Then $\lambda[d_1,\theta_2]$ and $\lambda[d_2,\theta_1]$ are $0$. Using the form of contract given in Proposition~\ref{prop3} and section~\ref{simpledescription} as well as the formula for the state-dependent transfer gives

\begin{align*}
\hat \gamma(d_1) = \frac{ \lambda[d_1,\theta_1]}{p(d_1)}, &\text{  } \beta(\theta_1) = \alpha y(d_1,\theta_1)+\hat \gamma(d_1) + \frac{ \lambda[d_1,\theta_1]}{p(d_1,\theta_1)}\\
\hat \gamma(d_2) = \frac{ \lambda[d_2,\theta_2]}{p(d_2)}, &\text{  } \beta(\theta_2) = \alpha y(d_2,\theta_2)+\hat \gamma(d_2) + \frac{  \lambda[d_2,\theta_2]}{p(d_2,\theta_2)}\\
\end{align*}
Since $\sum_{d}\lambda[d,\theta] = (1-\xi)\pi(\theta)$, we have $\lambda[d_1,\theta_1]=(1-\xi)\pi(\theta_1)$ and $\lambda[d_2,\theta_2]=(1-\xi)\pi(\theta_2)$ and so
\begin{align}
b(d_1,\theta_2) &= \alpha(10-5) +\frac{(1-\xi)\pi(\theta_2)}{p(d_2)}-\frac{(1-\xi)\pi(\theta_2)}{p(d_2,\theta_2)}-\frac{(1-\xi)\pi(\theta_1)}{p(d_1)}\label{opconShannon1}  \\ 
b(d_2,\theta_1) &= \alpha(5-0) +\frac{(1-\xi)\pi(\theta_1)}{p(d_1)}-\frac{(1-\xi)\pi(\theta_1)}{p(d_1,\theta_1)}-\frac{(1-\xi)\pi(\theta_2)}{p(d_2)} \label{opconShannon2}
\end{align}
\noindent Recall that the minimum agent reservation utility such that the optimal contract induces first-best experimentation is $2.853$. For this reservation utility the unique optimal contract is given by 
\begin{equation*}
\frac{1}{1+\mu}(y+\beta'),
\end{equation*}
where $\mu=0.446$ and $\beta'(\theta_1)=0$, $\beta'(\theta_2)=5$. Note that for a first-best contract we have $\xi=1$ since a fall in the agent's reservation utility induces a one-for-one increase in the principal's utility. Using $\xi=1$ in Equations (\ref{opconShannon1}) and (\ref{opconShannon2}) and setting $\alpha$ such that the cost of the optimal experiment is $k=0.5$, i.e. $\alpha=1/(1+0.446)$, gives the same contract as the first-best contract for reservation utility $2.853$.
To solve for this optimal contract in general requires solving the system of equations (\ref{Shannonexperiment}), (\ref{opconShannon1}), and (\ref{opconShannon2}) and the inequalities given by the agent's capacity constraint and participation constraint.

A special case arises when the prior is uniform. Then a solution to the system of equations is given by $p(d_1)=p(d_2)=1/2$, $p(d_1,\theta_1) = p(d_2,\theta_2)$, $b(d_1,\theta_2) = b(d_2,\theta_1)$ which implies that the optimal contract is affine and given by $b(d_1,\theta_2) = 5\alpha -t$, $b(d_2,\theta_t) = 5\alpha -t$, and $0$ otherwise, where $t = (1-\xi)/p(d_1|\theta_1) = (1-\xi)/p(d_2|\theta_2)$.

We now develop some graphical intuition for the second-best Pareto optimal contracts characterized in this example making the simplifying assumptions that the agent's capacity constraint does not bind---so that $\alpha^* = 1$---and that the agent's reservation utility is sufficiently low that his participation constraint does not bind---which implies $\xi=0$. Solving Equations (\ref{Shannonexperiment}), (\ref{opconShannon1}), and (\ref{opconShannon2}) then gives the following optimal contract, experiment, state-dependent transfer $\beta$, and decision-dependent transfer $\gamma$

\begin{table}[H]
\centering
\subfloat[]{
\begin{tabular}{ r !{\vrule width0.9pt} c  c }
 $b=y-\beta-\gamma$ &  $\theta_1$ &  $\theta_2$  \\ \Xhline{0.9pt}
 $d_1$  & 0 & 1.00   \\
 $d_2$  & 0.702 & 0  
\end{tabular}\label{optimalcontractsecondbest}
}
\hspace{5em}
\subfloat[]{\begin{tabular}{ r !{\vrule width0.9pt} c  c }
 $p(d|\theta)$ &  $\theta_1$ &  $\theta_2$  \\ \Xhline{0.9pt}
 $d_1$  & 0.160 & 0.514   \\
 $d_2$  & 0.840 & 0.486  
\end{tabular}\label{optimalexperimentfortheoptimecontract}}
\\
\vspace{2em}
\subfloat[]{
\begin{tabular}{ r !{\vrule width0.9pt} c  c }
 $\beta$ &  $\theta_1$ &  $\theta_2$  \\ \Xhline{0.9pt}
 $d_1$  & 3.836 & 6.596   \\
 $d_2$  & 3.836 & 6.596  
\end{tabular}\label{beta}
}
\hspace{5em}
\subfloat[]{\begin{tabular}{ r !{\vrule width0.9pt} c  c }
 $\gamma$ &  $\theta_1$ &  $\theta_2$  \\ \Xhline{0.9pt}
 $d_1$  & -3.836 & 2.404   \\
 $d_2$  & 0.462 & -1.596  
\end{tabular}\label{gamma}}\captionof{table}{\small \protect\subref{optimalcontractsecondbest} Optimal contract \protect\subref{optimalexperimentfortheoptimecontract} Optimal experiment \protect\subref{beta} State-dependent transfer \protect\subref{gamma} Optimal distortion}\label{Table2}
\end{table}
\noindent We will now build this optimal contract in stages. First we consider the contract $y-\beta$ shown in Table~\ref{first_best_contract_beta}. Since this contract is output adjusted by a state-dependent transfer the agent chooses a first-best experiment. Table~\ref{first_best_experiment_unconstrained1} presents the optimal experiment for this contract and Figure~\ref{firstbest} depicts the optimal posterior probabilities.
\begin{table}[H]
\centering
\subfloat[]{
\begin{tabular}{ r !{\vrule width0.9pt} c  c }
 $y-\beta$ &  $\theta_1$ &  $\theta_2$  \\ \Xhline{0.9pt}
 $d_1$  & -3.836 & 3.404   \\
 $d_2$  & 1.164 & -1.596  
\end{tabular}\label{first_best_contract_beta}
}
\hspace{5em}
\subfloat[]{
\begin{tabular}{ r !{\vrule width0.9pt} c  c }
 $p(d|\theta)$ &  $\theta_1$ &  $\theta_2$  \\ \Xhline{0.9pt}
 $d_1$  & 0.007   & 0.993   \\
 $d_2$  &  0.993 &  0.007
\end{tabular}\label{first_best_experiment_unconstrained1}
}\captionof{table}{\small \protect\subref{first_best_contract_beta} Contract $y-\beta$. \protect\subref{first_best_experiment_unconstrained1} First best experiment under contract $y-\beta$.}
\end{table}

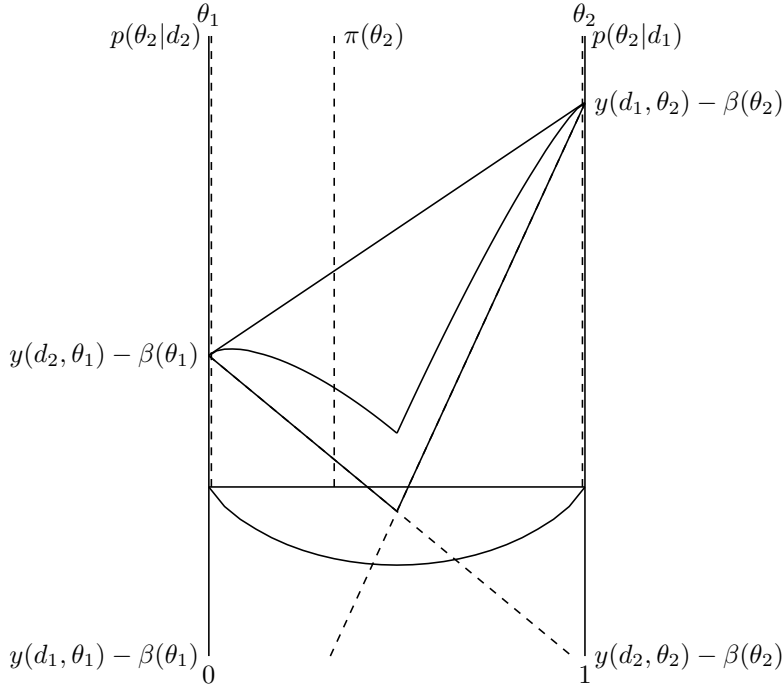
\begin{figure}[H]
\centering
\begin{tikzpicture}[xscale=5,yscale=1.5]
\draw[semithick,-] (0,0) -- (1,0)  ; 
\draw [semithick, -] (0,-1.5) node[below]{$0$} -- (0,4) node[above]{$\theta_1$}; 
\draw [semithick, -] (1,-1.5) node[below]{$1$}  -- (1,4) node[above]{$\theta_2$}; 
\draw [semithick, domain=0.0001:0.9999] plot (\x, {((\x)*ln(\x)+(1-\x)*ln(1-\x))}); 
\draw [dashed, semithick,-] (0.323,-1.5) -- (1,3.404) node[right] {$y(d_1,\theta_2)-\beta(\theta_2)$} ;
\draw [dashed, semithick,-] (0,-1.5)node[left] {$y(d_1,\theta_1)-\beta(\theta_1)$} ;
\draw [dashed, semithick,-] (0, 1.164) node[left] {$y(d_2,\theta_1)-\beta(\theta_1)$}-- (0.965,-1.5);
\draw [dashed, semithick,-] (1,-1.5) node[right] {$y(d_2,\theta_2)-\beta(\theta_2)$};
\draw [semithick, domain=0.0001:0.500155] plot (\x, {-2.7606*\x + 1.164}); 
\draw [semithick, domain=0.500155:1]   plot (\x, {7.244*\x - 3.8397});
\draw [semithick, domain=0.0001:0.500155]  plot (\x, {-((\x)*ln(\x)+(1-\x)*ln(1-\x))-2.7606*\x + 1.164});
\draw [semithick, domain=0.500155:1]  plot (\x, {-((\x)*ln(\x)+(1-\x)*ln(1-\x))+7.244*\x - 3.8397});
\draw [semithick, domain=0:1]  plot (\x, {2.24248*\x + 1.16853)});
\draw [semithick, dashed ,-] (0.00669168,0) -- (0.00669168,4) node[left] {$p(\theta_2|d_2)$};
\draw [semithick, dashed ,-] (0.993332,0)  -- (0.993332,4) node[right] {$p(\theta_2|d_1)$};
\draw [semithick, dashed ,-] (1/3,0) -- (1/3,4) node[right] {$\pi(\theta_2)$};
\end{tikzpicture}
\caption{\small Contract $y-\beta$ induces first best experimentation.}\label{firstbest}
\end{figure}
The contract $y-\beta$ induces very precise experimentation with the agent choosing the wrong decision with a $0.7$ percent probability in either state.

Now consider the feasible contract induced by $y-\beta$ which is its truncation $\max\{0,y-\beta\}$. Table~\ref{experiment for truncated contract} shows the optimal experiment for this contract and Figure~\ref{thirdbest} depicts the optimal posterior probabilities (thick lines) and compares it to the optimal posterior probabilities for $y-\beta$ (thin lines)

\begin{figure}[H]
\centering
\begin{tabular}{ r !{\vrule width0.9pt} c  c }
 $p(d|\theta)$ &  $\theta_1$ &  $\theta_2$  \\ \Xhline{0.9pt}
 $d_1$  & 0.211 & 0.963   \\
 $d_2$  & 0.789 & 0.037  
\end{tabular}\captionof{table}{\small Optimal experiment for contract $\max\{0,y-\beta\}$}\label{experiment for truncated contract}
\end{figure}

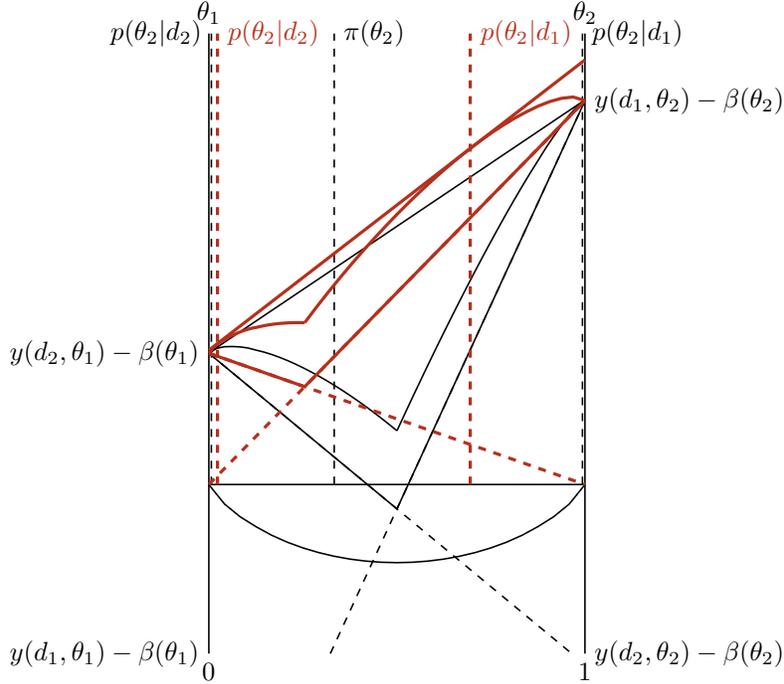
\begin{figure}[H]
\centering
\begin{tikzpicture}[xscale=5,yscale=1.5]
 \draw[semithick,-] (0,0) -- (1,0)  ; 
\draw [semithick, -] (0,-1.5) node[below]{$0$} -- (0,4) node[above]{$\theta_1$}; 
\draw [semithick, -] (1,-1.5) node[below]{$1$}  -- (1,4) node[above]{$\theta_2$}; 
\draw [semithick, domain=0.0001:0.9999] plot (\x, {((\x)*ln(\x)+(1-\x)*ln(1-\x))}); 
\draw [dashed, semithick,-] (0.323,-1.5) -- (1,3.404) node[right] {$y(d_1,\theta_2)-\beta(\theta_2)$} ;
\draw [dashed, semithick,-] (0,-1.5)node[left] {$y(d_1,\theta_1)-\beta(\theta_1)$} ;
\draw [dashed, semithick,-] (0, 1.164) node[left] {$y(d_2,\theta_1)-\beta(\theta_1)$}-- (0.965,-1.5);
\draw [dashed, semithick,-] (1,-1.5) node[right] {$y(d_2,\theta_2)-\beta(\theta_2)$};
\draw [semithick, domain=0.0001:0.500155] plot (\x, {-2.7606*\x + 1.164}); 
\draw [semithick, domain=0.500155:1]   plot (\x, {7.244*\x - 3.8397});
\draw [semithick, domain=0.0001:0.500155]  plot (\x, {-((\x)*ln(\x)+(1-\x)*ln(1-\x))-2.7606*\x + 1.164});
\draw [semithick, domain=0.500155:1]  plot (\x, {-((\x)*ln(\x)+(1-\x)*ln(1-\x))+7.244*\x - 3.8397});
\draw [semithick, domain=0:1]  plot (\x, {2.24248*\x + 1.16853)});
\draw [semithick, dashed ,-] (0.00669168,0) -- (0.00669168,4) node[left] {$p(\theta_2|d_2)$};
\draw [semithick, dashed ,-] (0.993332,0)  -- (0.993332,4) node[right] {$p(\theta_2|d_1)$};
\draw [semithick, dashed ,-] (1/3,0) -- (1/3,4) node[right] {$\pi(\theta_2)$};

\draw [BrickRed, dashed, very thick,-] (0,0) -- (1,3.404) ;
\draw [BrickRed, dashed, very thick,-] (0, 1.164) -- (1,0);
\draw [BrickRed, very thick, domain=0.0001:0.255]  plot (\x, {1.164 - (1.164)*\x});
\draw [BrickRed, very thick, domain=0.255:1]  plot (\x, {(3.404)*\x});
\draw [BrickRed, very thick, domain=0.0001:0.255]  plot (\x, {-((\x)*ln(\x)+(1-\x)*ln(1-\x))+1.1640 - (1.1640)*\x});
\draw [BrickRed, very thick, domain=0.255:1]  plot (\x, {-((\x)*ln(\x)+(1-\x)*ln(1-\x))+(3.404)*\x});
\draw [BrickRed, very thick, domain=0:1]  plot (\x, {2.575*\x + 1.191)});
\draw [BrickRed, very thick, dashed ,-] (0.023,0) -- (0.023,4) node[right] {$p(\theta_2|d_2)$};
\draw [BrickRed, very thick, dashed ,-] (0.695,0)  -- (0.695,4) node[right] {$p(\theta_2|d_1)$};
\end{tikzpicture}
\caption{\small Optimal experiments for $\max\{0,y-\beta\}$ (thick lines) and $y-\beta$ (thin lines).}\label{thirdbest}
\end{figure}

Under contract $\max\{0,y-\beta\}$, when the agent chooses decision $d_2$ he does so with approximately the same level of confidence in which state it is as under the first-best experiment. But when the agent chooses decision $d_1$ he does so with much less confidence about the state than under the first-best experiment. 

The optimal experiment for contract $\max\{0,y-\beta\}$ shown in Table~\ref{experiment for truncated contract} shows that in state $\theta_2$ the agent chooses the incorrect experiment $d_2$ with a similar probability as under the first best experiment. However the agent now incorrectly chooses decision $d_1$ in state $\theta_1$ about $20$ percent of the time compared to $2$ percent under the first best experiment. This effect arises because truncating the contract $y-\beta$ eliminates a greater downside for decision $d_1$ than $d_2$ and so encourages the agent to choose $d_1$. 

A way to measure the change in the principal's utility due to truncating the contract $y-\beta$ is shown in Figure~\ref{generalapproach}. Point $A$ gives the expected value of output less the constant $E_\pi[\beta(\theta)]$ and point $B$ gives the expected contract payment. Since under the contract $y-\beta$ these points coincide and now point $A$ is below $B$, the move from contract $y-\beta$ to $\max\{0,y-\beta\}$ reduces the principal's utility. 

\begin{figure}[H]
\centering
\begin{tikzpicture}[xscale=5,yscale=1.5]
 \draw[semithick,-] (0,0) -- (1,0)  ; 
\draw [semithick, -] (0,-1.5) node[below]{$0$} -- (0,4) node[above]{$\theta_1$}; 
\draw [semithick, -] (1,-1.5) node[below]{$1$}  -- (1,4) node[above]{$\theta_2$}; 
\draw [semithick, domain=0.0001:0.9999] plot (\x, {((\x)*ln(\x)+(1-\x)*ln(1-\x))}); 
\draw [dashed, semithick,-] (0.323,-1.5) -- (1,3.404) node[right] {$y(d_1,\theta_2)-\beta(\theta_2)$} ;
\draw [dashed, semithick,-] (0,-1.5)node[left] {$y(d_1,\theta_1)-\beta(\theta_1)$} ;
\draw [dashed, semithick,-] (0, 1.164) node[left] {$y(d_2,\theta_1)-\beta(\theta_1)$}-- (0.965,-1.5);
\draw [dashed, semithick,-] (1,-1.5) node[right] {$y(d_2,\theta_2)-\beta(\theta_2)$};
\draw [semithick, domain=0.0001:0.500155] plot (\x, {-2.7606*\x + 1.164}); 
\draw [semithick, domain=0.500155:1]   plot (\x, {7.244*\x - 3.8397});
\draw [semithick, domain=0.0001:0.500155]  plot (\x, {-((\x)*ln(\x)+(1-\x)*ln(1-\x))-2.7606*\x + 1.164});
\draw [semithick, domain=0.500155:1]  plot (\x, {-((\x)*ln(\x)+(1-\x)*ln(1-\x))+7.244*\x - 3.8397});
\draw [semithick, domain=0:1]  plot (\x, {2.24248*\x + 1.16853)});
\draw [semithick, dashed ,-] (0.00669168,0) -- (0.00669168,4) node[left] {$p(\theta_2|d_2)$};
\draw [semithick, dashed ,-] (0.993332,0)  -- (0.993332,4) node[right] {$p(\theta_2|d_1)$};
\draw [semithick, dashed ,-] (1/3,0) -- (1/3,4) node[right] {$\pi(\theta_2)$};
\draw [BrickRed, very thick, domain=0.0231:0.695]  plot (\x, {1.828*\x+1.095});

\draw [BrickRed, dashed, very thick,-] (0,0) -- (1,3.404) ;
\draw [BrickRed, dashed, very thick,-] (0, 1.164) -- (1,0);
\draw [BrickRed, very thick, domain=0.0001:0.255]  plot (\x, {1.164 - (1.164)*\x});
\draw [BrickRed, very thick, domain=0.255:1]  plot (\x, {(3.404)*\x});
\draw [BrickRed, very thick, domain=0.0001:0.255]  plot (\x, {-((\x)*ln(\x)+(1-\x)*ln(1-\x))+1.1640 - (1.1640)*\x});
\draw [BrickRed, very thick, domain=0.255:1]  plot (\x, {-((\x)*ln(\x)+(1-\x)*ln(1-\x))+(3.404)*\x});
\draw [BrickRed, very thick, domain=0:1]  plot (\x, {2.575*\x + 1.191)});
\draw [BrickRed, very thick, dashed ,-] (0.023,0) -- (0.023,4) node[right] {$p(\theta_2|d_2)$};
\draw [BrickRed, very thick, dashed ,-] (0.695,0)  -- (0.695,4) node[right] {$p(\theta_2|d_1)$};
\draw [BrickRed, very thick, domain=0.0231:0.695]  plot (\x, {1.828*\x+1.095});
\draw [black, semithick, domain=0.0231:0.695]  plot (\x, {0.139*\x+1.097});
\draw (1/3,1.704) node[right]{$B$};
\draw (1/3,1.704) node[circle,fill,BrickRed,draw,scale=.3]{};
\draw (1/3,1.144) node[right]{$A$};
\draw (1/3,1.144) node[circle,fill,draw,scale=.3]{};
\end{tikzpicture}
\caption{\small Point $A$ gives expected output minus $E_\pi[\beta(\theta)]$. Point $B$ gives the expected contract payment.}\label{generalapproach}
\end{figure}
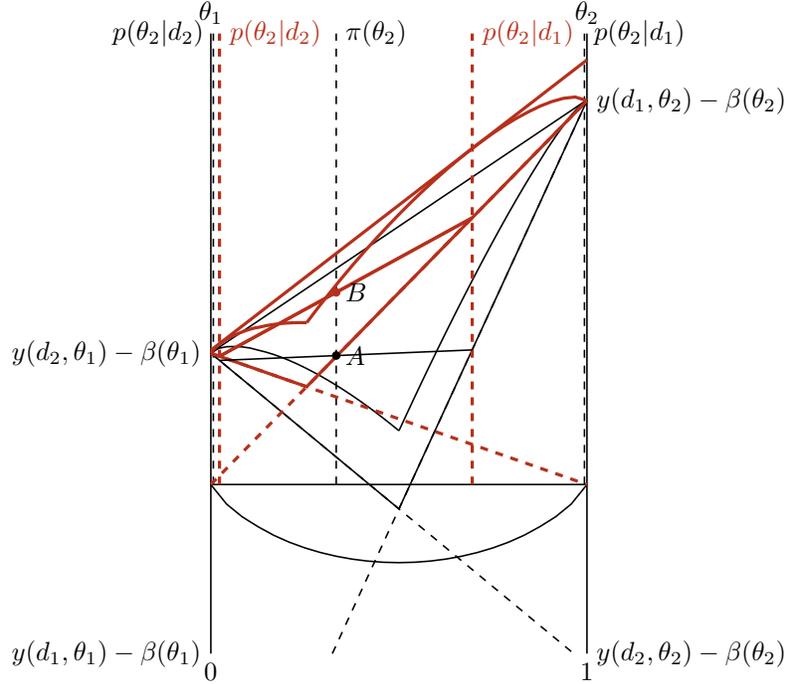

Figure~\ref{optimalcontract} shows the optimal posterior beliefs for the optimal contract $b=y-\beta-\gamma$. Table~\ref{optimalexperimentfortheoptimecontract} gives the optimal experiment and Table~\ref{table comparison of payoffs and costs} lists the expected value of output, the contract, and the experiment costs under the three different contracts.

\begin{figure}[H]
\centering
\begin{tabular}{ r !{\vrule width0.9pt} c  c c   }
  &  $E_p[y]$ &  $E_p[b]$ & $c(p)$   \\ \Xhline{0.9pt}
 $y-\beta$  & 6.633 & 1.877 & 0.596  \\
$\max\{0,y-\beta\}$& 5.900  & 1.704 & 0.293  \\
$b=y-\beta-\gamma$& 5.321  & 0.566 & 0.067  
\end{tabular}\captionof{table}{\small Comparison of expected payoffs and costs}\label{table comparison of payoffs and costs}
\end{figure}

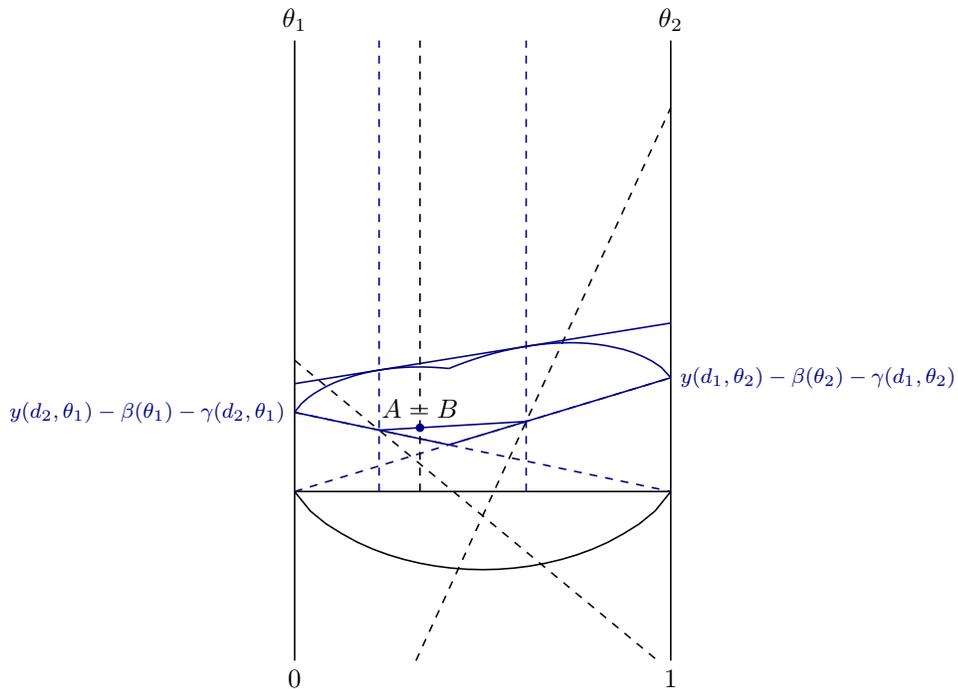
\begin{figure}[H]
\centering
\begin{tikzpicture}[xscale=5,yscale=1.5]
 \draw[semithick,-] (0,0) -- (1,0)  ; 
\draw [semithick, -] (0,-1.5) node[below]{$0$} -- (0,4) node[above]{$\theta_1$}; 
\draw [semithick, -] (1,-1.5) node[below]{$1$}  -- (1,4) node[above]{$\theta_2$}; 
\draw [semithick, domain=0.0001:0.9999] plot (\x, {((\x)*ln(\x)+(1-\x)*ln(1-\x))}); 
\draw [dashed, semithick,-] (0.322643,-1.5) -- (1,3.4042) ;
\draw [dashed, semithick,-] (0, 1.164)-- (0.965287,-1.5);
\draw [NavyBlue, dashed, semithick,-] (0,0)  -- (1,1.0089) node[right] {\footnotesize{$y(d_1,\theta_2)-\beta(\theta_2)-\gamma(d_1,\theta_2)$}} ;
\draw [NavyBlue, dashed, semithick,-] (0, 0.7019) node[left] {\footnotesize{$y(d_2,\theta_1)-\beta(\theta_1)-\gamma(d_2,\theta_1)$}}-- (1,0);
\draw [NavyBlue, semithick, domain=0.0001:0.410276]  plot (\x, {0.7019 - 0.7019*\x});
\draw [NavyBlue, semithick, domain=0.410276:1]  plot (\x, {1.0089*\x});
\draw [NavyBlue, semithick, domain=0.0001:0.410276]  plot (\x, {-((\x)*ln(\x)+(1-\x)*ln(1-\x))+0.7019 - 0.7019*\x});
\draw [NavyBlue, semithick, domain=0.410276:1]  plot (\x, {-((\x)*ln(\x)+(1-\x)*ln(1-\x))+1.0089*\x});
\draw [NavyBlue, semithick, domain=0:1]  plot (\x, {0.537939*\x+0.956101});
\draw [NavyBlue, semithick, dashed ,-] (0.2245,0) -- (0.2245,4);
\draw [NavyBlue, semithick, dashed ,-] (0.6156,0)  -- (0.6156,4);
\draw [semithick, dashed ,-] (1/3,0) -- (1/3,4);
\draw (1/3,0.565682) node[NavyBlue, circle,fill,draw,scale=.3]{};
\draw [NavyBlue, semithick, domain=0.2245:0.6156]  plot (\x, {0.196255*\x+0.500264});
\draw [black,semithick,-] (1/3,0.565682)  node[above] {$A=B$};
=\end{tikzpicture}
\caption{\small The optimal contract. The lines defining $A$ and $B$ coincide so the optimal contract $b=y-\beta-\gamma$ has the same expected value as the contract $y-\beta$ at the optimal posteriors.}\label{optimalcontract}
\end{figure}

You can see that this contract is locally optimal since the lines defining $A$ and $B$ coincide and so for small perturbations the change in the expected contract payment will be roughly equal to the change in expected output. This also implies that the expected value of the optimal contract at the optimal posterior beliefs for this contract coincide with the expected value of the first best contract $y-\beta$ at these posterior beliefs.


Figure~\ref{awayfrommax} shows a perturbation away from the optimal contract. You can see that the expected contract payment falls less than expected output so the perturbation is not profitable for the principal. 

\begin{figure}[H]
\centering
\begin{tikzpicture}[xscale=5,yscale=1.5]
\draw[semithick,-] (0,0) -- (1,0)  ; 
\draw [semithick, -] (0,-1.5) node[below]{$0$} -- (0,4) node[above]{$\theta_1$}; 
\draw [semithick, -] (1,-1.5) node[below]{$1$}  -- (1,4) node[above]{$\theta_2$}; 
\draw [semithick, domain=0.0001:0.9999] plot (\x, {((\x)*ln(\x)+(1-\x)*ln(1-\x))}); 
\draw [dashed, semithick,-] (0.322643,-1.5) -- (1,3.4042) ;
\draw [dashed, semithick,-] (0, 1.164)-- (0.965287,-1.5);
\draw [NavyBlue, dashed, semithick,-] (0,0)  -- (1,1.0089) node[right] {\footnotesize{$y(d_1,\theta_2)-\beta(\theta_2)-\gamma(d_1,\theta_2)$}} ;
\draw [NavyBlue, dashed, semithick,-] (0, 0.7019) node[left] {\footnotesize{$y(d_2,\theta_1)-\beta(\theta_1)-\gamma(d_2,\theta_1)$}}-- (1,0) ;
\draw [NavyBlue, semithick, domain=0.0001:0.410276]  plot (\x, {0.7019 - 0.7019*\x});
\draw [NavyBlue, semithick, domain=0.410276:1]  plot (\x, {1.0089*\x});
\draw [NavyBlue, semithick, domain=0.0001:0.410276]  plot (\x, {-((\x)*ln(\x)+(1-\x)*ln(1-\x))+0.7019 - 0.7019*\x});
\draw [NavyBlue, semithick, domain=0.410276:1]  plot (\x, {-((\x)*ln(\x)+(1-\x)*ln(1-\x))+1.0089*\x});
\draw [NavyBlue, semithick, domain=0:1]  plot (\x, {0.537939*\x+0.956101});
\draw [NavyBlue, semithick, dashed ,-] (0.2245,0) -- (0.2245,4);
\draw [NavyBlue, semithick, dashed ,-] (0.6156,0)  -- (0.6156,4);
\draw [NavyBlue, semithick, domain=0.2245:0.6156]  plot (\x, {0.196255*\x+0.500264});
\draw [NavyBlue, semithick, domain=0.2245:0.6156]  plot (\x, {0.196255*\x+0.500264});

\draw [BrickRed, dashed, very thick,-] (0,0) -- (1,0.85) ;
\draw [NavyBlue, dashed, semithick,-] (0, 0.7019)-- (1,0);
\draw [BrickRed, very thick, domain=0.0001:0.452284]  plot (\x, {0.7019 - 0.7019*\x});
\draw [BrickRed, very thick, domain=0.452284:1]  plot (\x, {0.85*\x});
\draw [BrickRed, very thick, domain=0.0001:0.452284]  plot (\x, {-((\x)*ln(\x)+(1-\x)*ln(1-\x))+0.7019 - 0.7019*\x});
\draw [BrickRed, very thick, domain=0.452284:1]  plot (\x, {-((\x)*ln(\x)+(1-\x)*ln(1-\x))+0.85*\x});
\draw [BrickRed, very thick, domain=0:1]  plot (\x, {0.274976*\x+1.02143});
\draw [BrickRed, very thick, dashed ,-] (0.273512,0) -- (0.273512,4);
\draw [BrickRed, very thick, dashed ,-] (0.639922,0)  -- (0.639922,4);
\draw [semithick, dashed ,-] (1/3,0) -- (1/3,4);
\draw [NavyBlue, semithick, domain=0.2245:0.6156]  plot (\x, {0.196255*\x+0.500264});
\draw [BrickRed, very thick ,-] (0.273512,0.509922) -- (0.639922,0.543933);
\draw [semithick ,-] (0.273512,0.409107) -- (0.639922,0.797032);
\draw (1/3,0.515475) node[circle,fill,BrickRed,draw,scale=.18]{};
\draw (1/3,0.472441) node[circle,fill,draw,scale=.18]{};
\draw (1/3,0.565682) node[NavyBlue,circle,fill,draw,scale=.18]{};
\end{tikzpicture}
\caption{\small A perturbation (thick line) away from the optimal contract. The expected contract payment (middle dot) falls less than expected output (lower dot).}\label{awayfrommax}
\end{figure}
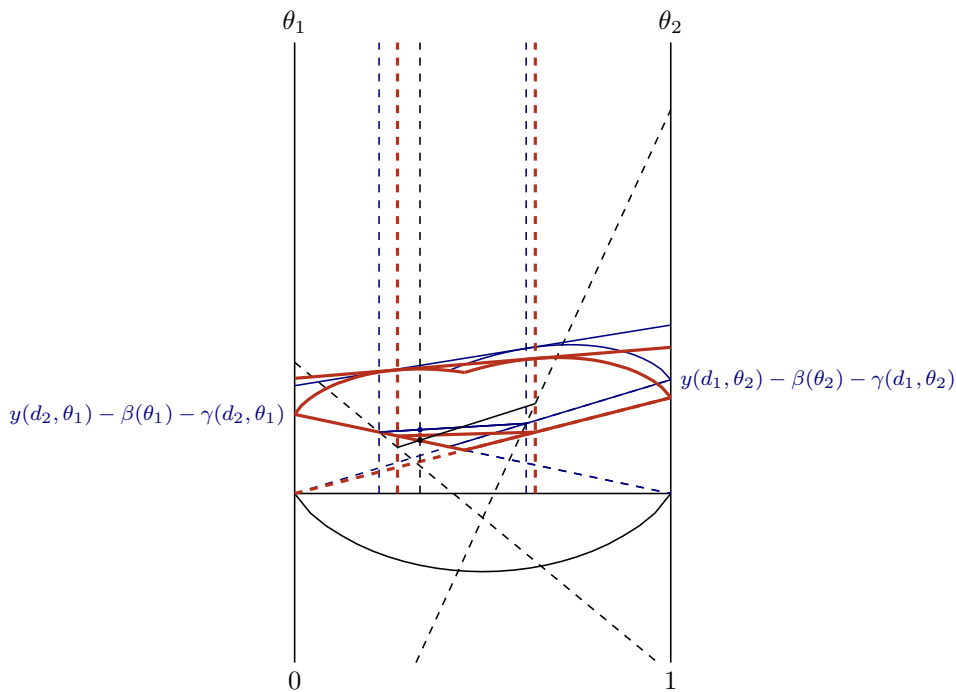

\section{Discussion}

We do not impose restrictions on contract form, such as monotonicity or linearity, or on the way the agent can influence the distribution over output, such as monotone likelihood ratio or distribution function convexity properties. Our model is at its core an information acquisition version of the static model in \cite{holmstrom1987aggregation} on linear contracts, in which the domain of the agent's cost function is a feasible set of distributions, rather than a set of effort parameters. This side-steps the need to specify a mapping between effort parameters and the distribution over output.

This approach also starts the analysis in the right place for thinking about linear contracts. A theme in that literature is that linear contracts arise because they are robust to uncertainty about the contracting environment and to freedom in the way the agent can influence the distribution over output (e.g.~\cite{carroll2015robustness,diamond1998managerial}). To see the second of these points, consider the following argument for why linear contracts align the incentives of a principal and agent: where it can be avoided non-linearity is wasteful because it drives apart the agent's and principal's preferred distributions over output; a convex payoff for the agent nearby some output level means a concave payoff for the principal, and thus an agent who wishes for higher variation in output than the principal nearby this level of output. This argument is strengthened by the agent having freedom in how he can influence the distribution over output. We take advantage of this by starting out with a model that allows the agent to choose any distribution over output. This starting point also forces us to think carefully about the restrictions we do place on the agent's feasible set, such as the capacity constraint.

Our results demonstrate the importance of capacity constraints in principal-agent models by showing that they pin down the fraction of output paid to the agent. In the context of modeling information acquisition, usual practice dictates that a decision maker either pays a cost for information in terms of utility, or faces a budget set of feasible experiments given by a capacity constraint. These approaches are interchangeable in decision problems---one is dual to the other. However, in principal-agent models the optimal contract varies significantly depending on whether or not a capacity constraint is imposed. 
The results we obtain regarding capacity constraints do not depend on the agent choosing an experiment. Rather, they apply to any standard principal-agent setup with a risk-neutral principal and agent. For example, one could augment the security design model of \cite{innes1990limited} with a capacity constraint and apply our results. The conclusion of that paper in which debt is optimal then changes, the optimal contract now being debt and a fraction of equity. This connects our work to the literature that seeks to derive optimal capital structure from agency considerations (e.g.~\cite{jensen1976theory}). 

\subsection{Debt and equity}

Our model describes interaction between an entrepreneur and investor in which the investor's role is to provide funding and the entrepreneur's job is to acquire information and take a decision. In the simplest case, when the agent's information cost is proportional to expected Shannon entropy reduction, our model predicts that
\[
y(d,\theta) - b(d,\theta)
\] 
\[= \min\{y(d,\theta),(\beta(\theta)+\hat \gamma(d))/\alpha^*\} + (1-\alpha^*)\max\{0,y(d,\theta)-(\beta(\theta)+\hat \gamma(d))/\alpha^*\}.
\]
Thus, the payoffs of the project are split into three pieces. A debt of $(\beta(\theta)+\hat \gamma(d))/\alpha^*$ and a $1-\alpha^*:\alpha^*$ split between investor and entrepreneur of equity $\max\{0,y(d,\theta)-(\beta(\theta)+\hat \gamma(d))/\alpha^*\}$.

\paragraph{Debt.} 
The debt's face value is
\[
\frac{\beta(\theta) + \hat \gamma(d)}{\alpha^*}.
\]
The state-dependent transfer $\beta$ corresponds to indexation of the debt (certain classes of bonds exhibit this feature, their coupon rate being a reference rate plus a quoted spread).
The decision-dependent transfer $\hat \gamma$ corresponds to a debt covenant, automatically increasing debts following a course of action that risks bankruptcy.

\paragraph{Equity.} 

Debt is repaid first. What remains is equity, which is split between the investor and entrepreneur. Fraction $\alpha^*$ remains inside the firm to motivate work; additional equity kept inside has no effect since capacity constraints bind, so it used outside the firm to raise capital.

\paragraph{Debt vs. equity.} 
A capacity constrained entrepreneur first sells equity since it does not distort his inventives. After a certain amount of equity is issued his capacity constraint will stop binding. If, before this point, the investor receives his required return, entrepreneur effort is first-best. Debt is issued only if more capital is needed.

\subsection{Restricted investment contracts}

\cite{Carroll16} studies a principal-agent problem of information acquisition where the principal does not know all experiments or experiment costs available to the agent. His principal ranks contracts according to their minimum expected payoff among all experiments and experiment costs including a known set. He shows that the optimal contract is a restricted investment contract: the set of decisions available to the agent is restricted and for unrestricted decisions the agent is paid a fraction of output less a state-dependent transfer.

Carroll's result assumes that for each convex combination of decisions there is a decision that dominates the combination: for each $t\in [0,1]$ and $d,d'\in D$ there is a decision $d''$ such that $y(d'',\theta)\geq t y(d,\theta) + (1-t)y(d',\theta)$ for all $\theta\in \Theta$. With a finite number of decisions this assumption implies that there is a dominant decision: a decision $d'$ such that $y(d',\theta)\geq y(d,\theta)$ for all $d\in D,\theta\in \Theta$. This eliminates the need for information acquisition. Carroll suggests that the set of decisions may be replaced by convex combinations of decisions. This would be problematic in our setting because a randomization over decisions is not verifiable (which for a restricted investment contract would be important to determine if the selected decision is restricted). Thus in Carroll's model the agent makes a recommendation and the principal is left to take the decision.

To compare our results to Carroll's, it is useful to think of his restriction on decisions as a payment $ \gamma:D\times \Theta \to \BR$ that sufficiently punishes restricted decisions that they are left unchosen and is zero for unrestricted decisions. Thus, Carroll obtains the contract $b = \alpha y-\beta- \gamma $ where $\alpha \in [0,1]$, $\beta:\Theta \to \BR$ is defined as $\beta(\theta) = \min\{\alpha y(d,\theta) - \gamma(d,\theta):d\in D_{UR}\}$, and $\gamma: D\times \Theta \to \BR$ is defined to be $0$ on unrestricted decisions $D_{UR}$ and such that $b(d,\theta)=0$ for restricted decisions $D\setminus D_{UR}$. In contrast to this, we obtain the contract $b=\alpha y-\beta-\gamma$ where $\alpha= \alpha^*\in [0,1]$ (as defined in Proposition~\ref{capacitythm}), $\beta:\Theta \to \BR$ is defined as $\beta(\theta) = \min\{\alpha y(d,\theta)-\gamma(d,\theta):d\in D\}$, and $\gamma:D\times \Theta\to \BR$ is given by
\begin{align*}
\gamma(d,\theta) = \frac{1}{\pi(\theta)}\sum_{ d',\theta' }\frac{\partial^2 c(p)}{\partial p(d|\theta)\partial p(d'|\theta')}\left((1-\xi)p(d'|\theta') - \frac{\lambda[d',\theta'] }{\pi(\theta')}\right),
\end{align*}
where $\lambda[d,\theta]$ and $\xi$ are dual variables for the liability limits $0\leq b(d,\theta)\leq y(d,\theta)$ and participation constraint. When our cost function is proportional to expected Shannon entropy reduction, $\gamma(d,\theta) = \sum_{\theta'}\lambda[d,\theta']/p(d)-\lambda[d,\theta]/p(d|\theta)\pi(\theta)$ which conditional on liability limits not binding does not depend on the state.

In both sets of results $\gamma$ punishes decisions that are likely to induce liability limits to bind. In Carroll's model, $\gamma$ either rules out decisions or does nothing. In our model $\gamma$ is less absolute: decisions are punished in proportion to how often and at what cost they induce binding liability limits. In both models the state-dependent transfers serve a normalizing role, setting the minimum contract payment in each state to zero.

Finally, the strength of the agent's incentives is measured by $\alpha$ in both Carroll's results and our own. Carroll does not discuss comparative statics for this parameter, but we are able to show in our model that it is increasing in the agent's capacity. It would be interesting to see whether a similar comparative static holds in Carroll's model, perhaps in terms of enlargements in the set of experiments and costs known to be available to the agent.

\subsection{Other related literature}

Acquiring information may be socially destructive. Effort spent may be wasteful, as in \cite{hirshleifer1971private} and \cite{spence1973job}, or it may lead to trade-destroying information asymmetry, as in \cite{akerlof1970market} and \cite{rothschild1976equilibrium}. These forces affect contract design. For example, in a model where an agent must report a project's outcome, \cite{townsend1979optimal} finds that debt contracts are optimal because they efficiently ration monitoring costs. Motivated by the second force, the same conclusion is drawn by \cite{dang2012ignorance} and \cite{yang2015optimality}: debt contracts are optimal to securitize state-dependent cash-flows because they avoid trade-destroying information asymmetry between the buyer and seller. Other examples include~\cite{cremer1992gathering}, \cite{cremer1998strategic}, \cite{yang2015financing}.

Recently, a theoretical literature has emerged that studies dynamic contracts that incentivize experimentation (e.g. \cite{bergemann1998venture}, \cite{gerardi2012principal}, \cite{guo2016dynamic}, \cite{halac2016optimal}, \cite{horner2013incentives}, \cite{manso2011motivating}). This literature often uses results from the strategic experimentation literature (e.g.~\cite{bolton1999strategic}, \cite{keller2005strategic}) but typically specializes to the case of a one or two armed bandit problem. An exception is \cite{chassang2013calibrated} who considers a very general model of dynamic information acquisition. He  shows that limited liability becomes approximately non-binding with a long enough time horizon and he uses this to characterize a class of contracts that work well in a variety of environments.

\cite{sims2003implications} and \cite{sims2005rational} pioneered the use of expected reduction in Shannon entropy as a cost function to model bounded rationality; its use can also be found in~\cite{arrow1973information}, \cite{arrow1985informational}, and in the statistical decision theory literature (e.g. \cite{degroot1962uncertainty}, \cite{lindley1956measure}). Recent efforts have aimed to derive behavioral foundations for different forms of costly information acquisition and to test popular cost functions using state-dependent stochastic choice data obtained from experiments (e.g.~\cite{caplin2013behavioral}, \cite{caplin2015revealed}, \cite{HebertWoodford16}, \cite{oliveira2017rationally}, \cite{woodford2012inattentive}).

\section{Conclusion}\label{Summary}

We have studied contracting for the production of an experiment and provided necessary and sufficient conditions for optimal contracts to take an intuitive linear form. All Pareto optimal contracts pay a fraction of output, a state-dependent transfer, and---with information costs proportional to expected reduction in Shannon entropy---a decision-dependent transfer. For a general cost function, our characterization shows how complementarities in the cost of receiving different signals in different states affects the design of the optimal contract. Our analysis provides explanations and comparative statics for the strength of incentives, the role of indexing, and the role of punishments and rewards.


\newpage

\begin{appendices}

\section{Results used in propositions~\ref{prop2} and~\ref{prop3}}

\begin{dfn}
(Tangent vector). A vector $w\in \BR^{n}$ is tangent to a set $C \subseteq \BR^n$ at a point $\bar x \in C$, written $w \in T_C(\bar x)$, if there exists a sequence $x_1,x_2,x_3,\cdots$ in $C$ converging to $\bar x$ along with a sequence of positive scalars $\tau_1,\tau_2,\tau_3,\cdots$ converging to $0$, such that the sequence $(x_1-\bar x)/\tau_1,(x_2-\bar x)/\tau_2,(x_3-\bar x)/\tau_3,\cdots$ converges to $w$.
\end{dfn}

\begin{dfn}
(Normal vector). Let $C \subseteq \BR^n$ and $\bar x\in C$. A vector $v$ is normal to $C$ at $\bar x$ in the regular sense, written $v\in \hat N_C(\bar x)$, if

\[v\cdot w \leq 0 \text{ for all } w\in T_{C}(\bar x).\]

It is normal to $C$ (in the general sense), written $v\in N_C(\bar x)$, if there is a sequence $x_1,x_2,x_3,\cdots$ in $C$ converging to $\bar x$, and a sequence $v_1,v_2,v_3,\cdots$ in $\hat N_C(x_1),\hat N_C(x_2),\hat N_C(x_3), \cdots$ converging to $v$.
\end{dfn}

\begin{dfn}\label{Clarke}
(Clarke regularity of sets). A closed set $C\subseteq \BR^n$ is regular at one of its points $\bar x$ in the sense of Clarke if every normal vector to $C$ at $\bar x$ is a regular normal vector, i.e. $N_C(\bar x) = \hat N_C(\bar x)$.
\end{dfn}

\begin{thm}\label{wetrock6.12} (6.12. of \cite{rockafellar2009variational}) (Basic first order conditions for optimality).
Consider a problem of maximizing a differentiable function $f_0$ over a set $C\subseteq \BR^n$. A necessary condition for $\bar x$ to be locally optimal is

\[\nabla f_0(\bar x) \in \hat N_C(\bar x).\]

When $C$ is convex and $f_0$ concave, this condition is sufficient for $\bar x$ to be globally optimal.
\end{thm}

\begin{thm}\label{wetrock6.14} (6.14 of \cite{rockafellar2009variational}) (Normal cones to sets with a constraint structure).
Let

\[C = \{x\in A : F(x) \in B\}\]

for closed sets $A\subseteq \BR^n$ and $B\subseteq \BR^m$ and a $\mathcal{C}^1$ mapping $F:\BR^n \to \BR^m$, written component-wise as $F(x) = (f_1(x),\cdots,f_m(x))$. Suppose that the following assumption, to be called the standard constraint qualification at $\bar x$, is satisfied:

\[
\begin{cases}
\text{the only vector } y\in N_B(F(\bar x)) \text{ for which}\\
-(y_1 \nabla f_1(\bar x) + \cdots + y_m \nabla f_m(\bar x)) \in N_A(\bar x) \text{ is } y = (0,\cdots,0).
\end{cases}
\]

Then, if $A$ is regular at $\bar x$ and $B$ is regular at $F(\bar x)$, $C$ is regular at $\bar x$ and

\[N_C(\bar x) = \{y_1 \nabla f_1(\bar x) + \cdots + y_m \nabla f_m(\bar x) + z : y \in N_B(F(\bar x)), z\in N_A(\bar x)\}.\]

\end{thm}

\begin{thm}\label{wetrock6.15} (6.15 of \cite{rockafellar2009variational}) (Normals to boxes).
Suppose $B = B_1\times \cdots \times B_n$, where each $B_j$ is a closed interval in $\BR$. Then $B$ is regular at every one of its points $\bar x = (\bar x_1,\cdots,\bar x_n)$, and its normal cones have the form $N_B(\bar x) = N_{B_1}(\bar x_1) \times \cdots \times N_{B_n}(\bar x_n)$, where

\[
N_{B_j}(\bar x_j) =
\begin{cases}
[0,\infty) & \text{ if } \bar x_j \text{ is only the left endpoint of } B_j, \\
(-\infty,0] & \text{ if } \bar x_j \text{ is only the right endpoint of } B_j, \\
\{0\} & \text{ if } \bar x_j \text{ is an interior point of } B_j, \\
(-\infty,\infty) & \text{ if } B_j \text{ is a one point interval}.
\end{cases}
\]

\end{thm}

\subsection{Risk-aversion}\label{r-aversion}

Here we consider the case where the agent is risk-averse with a concave Bernoulli utility function $u$ over wealth. When the agent is not capacity constrained the formula for the optimal contract extends straightforwardly. The only change is that $\gamma(d,\theta)$ is now given by
\[
\frac{1}{\pi(\theta)}\sum_{ d',\theta' }\frac{1}{u'\left(b(d,\theta')\right)}\left(p(d'|\theta')(1-u'\left(b(d,\theta')\right)\xi) - \frac{\lambda[d',\theta']}{\pi(\theta')}\right)\frac{\partial^2 c(p)}{\partial p(d|\theta)\partial p(d'|\theta')}.
\]
For the cost functions of~\cite{HebertWoodford16} this becomes
\[
\gamma(d,\theta)= \sum_{\theta'}\left(\frac{k(\theta,\theta',p(\cdot|d))}{p(\theta|d)p(\theta'|d)}\right)\left(\frac{p(d,\theta') -\lambda[d,\theta']}{u'\left(b(d,\theta')\right)}\right).
\]
In this case, $\gamma$ takes on an insurance role as well as countering undesirable incentives for risk-taking induced by limited liability.

\end{appendices}

\end{document}